\documentclass[10pt,aps,pra,showpacs,twocolumn]{revtex4-1}
\usepackage{amsmath, amssymb}
\usepackage{amsthm}
\usepackage{graphicx}
\usepackage{xcolor}
\usepackage{bm}
\usepackage{subcaption}
\usepackage{algpseudocode}
\usepackage{hyperref}

\newtheorem{lmn}{Lemma}
\newtheorem{thrm}{Theorem}

\renewcommand{\vec}[1]{\mathbf{#1}}

\makeatletter
\renewcommand*\env@matrix[1][*\c@MaxMatrixCols c]{%
   \hskip -\arraycolsep
   \let\@ifnextchar\new@ifnextchar
   \array{#1}}
\makeatother

\begin{document}
	
\title{On the waiting time in quantum repeaters with probabilistic entanglement swapping}
\author{E. Shchukin}
\email{evgeny.shchukin@gmail.com}
\author{F. Schmidt}
\email{fschmi@students.uni-mainz.de}
\author{P. van Loock}
\email{loock@uni-mainz.de}
\affiliation{Johannes-Gutenberg University of Mainz, Institute of Physics, Staudingerweg 7, 55128 Mainz}

\begin{abstract}
The standard approach to realize a quantum repeater relies upon probabilistic but heralded entangled state manipulations
and the storage of quantum states while waiting for successful events. In the literature on this class of repeaters,
calculating repeater rates has typically depended on approximations assuming sufficiently small probabilities. Here we
propose an exact and systematic approach including an algorithm based on Markov chain theory to compute the average
waiting time (and hence the transmission rates) of quantum repeaters with arbitrary numbers of links. For up to four
repeater segments, we explicitly give the exact rate formulae for arbitrary entanglement swapping probabilities.
Starting with three segments, we explore schemes with arbitrary (not only doubling) and dynamical (not only
predetermined) connections. The effect of finite memory times is also considered and the relative influence of the
classical communication (of heralded signals) is shown to grow significantly for larger probabilities. Conversely, we
demonstrate that for small swapping probabilities the statistical behavior of the waiting time in a quantum repeater
cannot be characterized by its average value alone and additional statistical quantifiers are needed. For large repeater
systems, we propose a recursive approach based on exactly but still efficiently computable waiting times of sufficiently
small sub-repeaters. This approach leads to better lower bounds on repeater rates compared to existing schemes.
\end{abstract}
    

\maketitle

\section{Introduction}

Quantum repeaters are essential ingredients for large-scale, fiber-based quantum networks because of the exponential
decay of photonic quantum information along the optical communication channels. Besides recent all-optical approaches
relying on experimentally still demanding quantum error correction procedures, memory-based quantum repeaters remain
good candidates to realize long-range quantum communication \cite{nphotonics.10.381, PhysRevLett.119.170502,
nphys.11.695, fqm, ncomm.8.15359, nphys.14.50, ncomm.9.363, PhysRevLett.120.183602, arXiv.1808.05393}. In such quantum
repeaters, probabilistic events like, especially, the heralded distribution of entangled Bell states over sufficiently
small elementary channel segments become independent through the use of sufficiently long-lasting quantum memories, thus
circumventing the exponential scaling when the segments are connected. Remarkably, the precise average waiting time, and
hence the communication rate in a non-deterministic quantum repeater is unknown and computing this time has been a
long-standing problem \cite{RevModPhys.83.33}. An exact analytic solution only exists for the special case of
deterministic entanglement swapping; otherwise any existing work on quantum repeaters relies upon approximations or
numerical simulations. In our work, we present a systematic scheme, based on a Markov formalism\footnote{For earlier
treatments of quantum repeaters and distributed quantum computation using the formalism of Markov chains, see
\cite{10.1117.12.811880} and \cite{PhysRevA.76.062323}, respectively.} and linear equation solving, that enables one to
obtain the exact waiting times in practically relevant regimes. For a repeater as large as 1024 segments, we demonstrate
that the trade-off between the computational efficiency and the prediction accuracy can be dealt with in a modular
approach by recursively applying exactly but still efficiently computable sub-chains like 32 times 32 segments. This
provides increasingly more accurate but still computable lower bounds on the repeater rates. Furthermore, our method
allows one to easily compute the full probability distribution of the random waiting time for non-deterministic quantum
repeaters including those with arbitrary (not only doubling) and dynamical (not only predetermined) connections and
including classical communication and finite memory effects.  

Quantum repeaters enable one, in principle, to extend (optical-fiber-based) quantum communication schemes such as
quantum cryptography to distances as large as 1000 km and beyond despite channel losses that typically increase
exponentially with distance. Thanks to some recent results on the bounds for point-to-point communication
\cite{Nature.Comm.5.5235, Nature.Comm.8.15043}, there are now well-defined benchmarks for long-distance quantum
communication. These bounds can be explored in terms of a secret key rate per mode (per channel use) and, when ignoring
all imperfections besides channel losses, this corresponds to an optimal raw qubit transmission rate without
intermediate stations like in a quantum repeater \footnote{A brief chronological summary of the different contributions
to the secret-key agreement capacity and the corresponding bounds can be found in Ref.~\cite{arXiv.1708.07142}}.

Since the standard approach to quantum repeaters is based on quantum memories and on (at least partially) probabilistic
operations on entangled states (distribution, swapping, and purification) \cite{PhysRevLett.81.5932, *[{for recent
alternative approaches that do not require quantum memories and are entirely based on quantum error correction codes
(so-called third-generation repeaters), see }] Scientific.Reports.6.20463, *[{for yet another scheme that is based upon
a deterministic method of entanglement purification, see }] PhysRevLett.120.030503}, some recent proposals consider
small-scale versions of quantum repeaters with a minimal number of memory stations (repeater links) \cite{Luong2016,
PhysRevA.96.052313, 2058-9565-3-3-034002, PhysRevA.89.012303}. While the rate analysis for the smallest repeater with
only two segments and one link is fairly straightforward, quantum repeaters with two or more links become increasingly
complex to analyze when the entanglement swapping is probabilistic. Indeed there is no explicit expression for the
average repeater waiting time in the literature for such advanced cases \cite{*[{In }] [{, methods to calculate rates in
memory-based quantum repeaters have been presented with a particular focus on imperfect quantum memories.}] phd}.
However, entanglement swapping based on heralded but non-deterministic Bell measurements is a rather natural situation,
especially in one of the most prominent approaches based on atomic ensembles and linear optics \cite{Nature.414.413}
where normally even the ideal (photonic) Bell measurement cannot exceed an efficiency of 1/2 \cite{Calsamiglia2001}.
Thus, so far, the typical approximate rate formulas that have been applied depended on the assumption of sufficiently
small probabilities \cite{RevModPhys.83.33, PhysRevA.92.012307, PhysRevA.87.052315, PhysRevA.87.062335}. Nonetheless,
the most efficient memory-based quantum repeater schemes would rely on high swapping probabilities, for instance, based
on suitable atom-light interactions \cite{Uphoff2016} or enhanced linear-optics Bell measurements
\cite{PhysRevA.84.042331, PhysRevLett.110.260501, PhysRevLett.113.140403}. While exact analytical rate formulas are
known only for the extreme case of fully deterministic entanglement swapping \cite{PhysRevA.83.012323}, in the present
work, we will address the entire range of arbitrary swapping probabilities including a full statistical analysis beyond
only average values. For large-scale repeaters, we propose to divide the whole system into smaller sub-repeater chains,
still sufficiently big to maintain the great accuracy of our formalism and sufficiently small to be efficiently
computable. Errors that occur during the long-distance entanglement distribution in a quantum repeater can be included
into the Markov-chain formalism. Here the focus will be on the distribution times, which, for example, in the context of
quantum key distribution would be related to a quantum repeater raw rate \footnote{In order to calculate a secret key
rate in quantum repeater-based communication, besides the raw rate, inclusion of errors requires determining the quantum
bit error rate. Vinay and Kok \cite{PhysRevA.99.042313}, building upon an earlier version of our work, explicitly showed
already for a particular repeater protocol how to incorporate errors into the Markov-chain approach}.

The paper is structured as follows. In Sec.~II we introduce the general setting and review existing results on the rate
analysis in quantum repeaters. Sec.~III then discusses some commonly used approximations and Sec.~IV gives a detailed
introduction into the Markov-chain formalism for quantum repeaters. Explicit examples of small repeaters and some
special cases are presented in Sec.~V. How to incorporate the effects of finite memory and classical communication times
into our rate analysis is described in Secs.~VII and VIII, respectively. An alternative, complementary approach based on
generating functions is introduced in Sec. IX and a possible way to compress the Markov chain describing a quantum
repeater system in order to make the analysis more efficient is discussed in Sec.~X on lumpability. Finally, Sec.~XI
briefly describes a numerical validation of some of our analytical results, Sec. XII treats larger repeater systems
based on a recursive application of our exact rate formulas, and Sec.~XIII concludes the paper.

\section{General setting and known results}

\begin{figure}
    \includegraphics{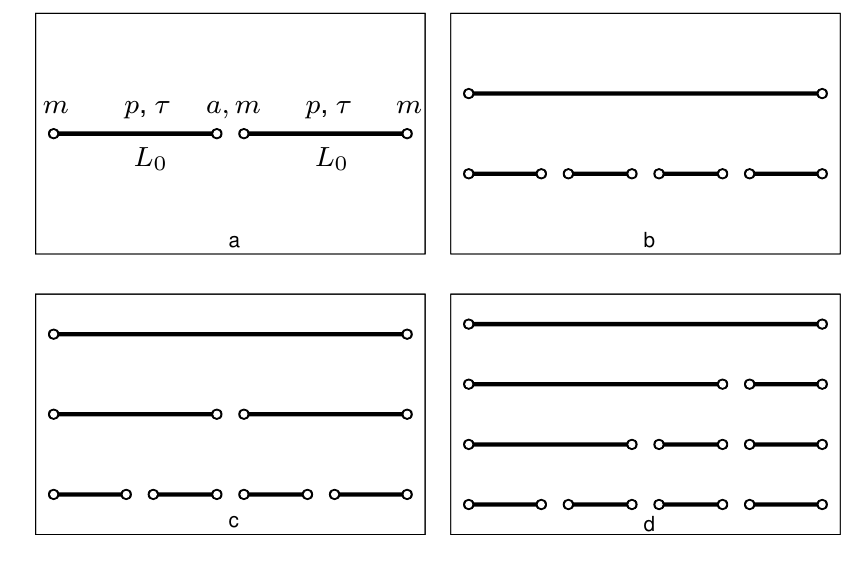}
    \caption{Different ways for entanglement distribution in a quantum repeater. Fig.~{\fontfamily{phv}\selectfont a}
    illustrates the basic parameters of our model (see main text). Figs. b-d describe various distinct situations for
    the entanglement swapping (also see main text).}\label{fig:qr}
\end{figure}

Let us consider a quantum repeater with $n$ identical segments of length $L_0$ and $n-1$ swappings between them. The
initial entanglement distribution success probability is denoted by $p$ and the swapping success probability by $a$. The
former includes the local state preparation efficiency and the channel transmission probability that decays
exponentially with distance. The latter may include local losses such as (heralded) memory erasures. The characteristic
time unit $\tau = L_0/c$ is the typical duration of a single distribution attempt (including the transmission times of
quantum and classical signals between neighboring repeater stations, and more specifically depending on the particular
repeater protocol). To distribute entanglement along the whole repeater we first distribute it in individual segments
and then combine them with swapping. If a swapping fails, or if a state is kept in memory longer than $\sim m$ time
units $\tau$, the affected segments are zeroed out and have to start entanglement distribution from the very beginning.
In Fig.~\ref{fig:qr}a we show two segments of such a quantum repeater. Figs.~\ref{fig:qr}b-\ref{fig:qr}c illustrate
schemes based on deterministic and probabilistic entanglement swapping, respectively, while the latter has the fixed
rule of doubling the distance on every nesting level. Fig.~\ref{fig:qr}d is more general, no longer imposing such rules.

Generally, with probability $a^{n-1} p^n$ we succeed in one step with no gain over direct state transmission. In order
to gain efficiency and change the scaling, we have to exploit quantum memories and allow to perform several steps to
successfully distribute entanglement over the whole repeater. We introduce the random variable $K_n$ which is equal to
the number of steps taken to successful distribution over the total distance $L = n L_0$ and, when multiplied with
$\tau$, corresponds to the total waiting time (on higher nesting levels different time units may apply, see Sec.~VIII).
The goal of this work is to study statistical properties of $K_n$. 

The main characteristic of any random variable is its probability distribution function (PDF). Up to now hardly anything
has been known about the full PDF of $K_n$ in the most general setting with arbitrary (including still relevant) $p$ and
$a$, although the waiting time in a single segment is obviously geometrically distributed with probability $pq^{k-1}$
for success at the $k$th attempt (where $q = 1-p$). A few results about the average $\overline{K}_n$ have been obtained.
The expression for the average waiting time of a two-segment repeater with non-deterministic swapping and ``memory
cutoff'' of $m$ time units has been obtained in Ref.~\cite{PhysRevLett.98.060502} and reads as
\begin{equation}\label{eq:Kfm}
    \overline{K}^{(m)}_2 = \frac{1+2q - 2q^{m+1}}{ap(1+q - 2q^{m+1})}.
\end{equation} 
In the case of deterministic swapping, $a=1$, and infinite memory, $m = +\infty$, the exact expression for
$\overline{K}_n$ for arbitrary $n$ was presented in Ref.~\cite{PhysRevA.83.012323} (Fig.~\ref{fig:qr}b illustrates such
a repeater; since swappings cannot fail, there are no additional levels). It reads as
\begin{equation}\label{eq:k1}
    \overline{K}_n = \sum^n_{j=1} (-1)^{j+1} \binom{n}{j} \frac{1}{1-q^j}.
\end{equation}
The expression for the average waiting time of an $n$-segment quantum repeater with deterministic swapping and a finite
memory cutoff of $m$ time units has been obtained in \cite{arXiv.1309.3407} and it is given by the following expression:
\begin{equation}\label{eq:Knm}
    \overline{K}^{(m)}_n = \frac{1-(1-q^m)^n + (1-q^n)\left(m - \sum\limits^{m-1}_{i=1}(1-q^i)^n\right)}{(1-q^{m+1})^n - q^n(1-q^m)^n}.
\end{equation}
We thus have three expressions independently obtained by different methods, so we need to verify that they are in
agreement where applicable. Simple algebra shows that the expression \eqref{eq:k1} for $n=2$ coincides with the
expression \eqref{eq:Kfm} for $a=1$ and $m=+\infty$ and that the expression \eqref{eq:Kfm} for $a=1$ coincides with the
expression \eqref{eq:Knm} for $n=2$. Verification that the expression \eqref{eq:Knm} for $m=+\infty$ coincides with the
expression \eqref{eq:k1} requires slightly more work.
\begin{lmn}
    The expressions \eqref{eq:k1} and \eqref{eq:Knm} are consistent for the case of infinite memory cutoff:
    $\overline{K}^{(+\infty)}_n = \overline{K}_n$.
\end{lmn}
\begin{proof}
We need only verify the following relation:
\begin{equation}
    \lim_{m \to +\infty} \left(m - \sum\limits^{m-1}_{j=1}(1-q^j)^n\right) = \overline{K}_n.
\end{equation}
The sum can be expanded as follows:
\begin{equation}
\begin{split}
    \sum^{m-1}_{i=1} (1-q^i)^n &= \sum^n_{j=0} (-1)^j \binom{n}{j} \sum^{m-1}_{i=1} q^{ij} \\
    &= m-1 + \sum^n_{j=1} (-1)^j \binom{n}{j} \frac{q^j - q^{jm}}{1-q^j}.
\end{split}
\end{equation}
We thus have 
\begin{equation}
\begin{split}
    &\lim_{m \to +\infty} \left(m - \sum\limits^{m-1}_{j=1}(1-q^j)^n\right) \\
    &= 1 + \sum^n_{j=1} (-1)^{j+1} \binom{n}{j} \frac{q^j}{1-q^j} = \overline{K}_n.
\end{split}
\end{equation}
In the last equality we used the relations
\begin{equation}
    \frac{q^j}{1-q^j} = \frac{1}{1-q^j} - 1
\end{equation}
and
\begin{equation}
    1 - \sum^n_{j=1} (-1)^{j+1} \binom{n}{j} = (1-1)^n = 0.
\end{equation}
This concludes the proof that the expressions \eqref{eq:k1} and \eqref{eq:Knm} are consistent.
\end{proof}

One must be careful by using the expressions \eqref{eq:k1} and \eqref{eq:Knm} for numerical evaluation of the averaged
waiting time. For large $n$ (roughly speaking, for $n>60$) and for small $p$ (roughly, for $p < 0.01$) one can obtain
wrong results due to the limitations of the standard double precision floating point arithmetic. In this case, either
one has to use multiple precision numbers or use the approximations presented in the next section.

\section{Approximations}

No expression for the exact average waiting time time is known in the case of non-deterministic swapping, $a<1$, and for
more than two segments, $n>2$. Moreover, for non-deterministic swapping different orders of the swapping operations lead
to schemes of significantly different efficiencies. One of these schemes is when the number of segments is a power of
two, $n = 2^d$, recursively doubling the entangled segments in the repeater (see Fig.~\ref{fig:qr}c). For this scheme,
an often used approximation \cite{RevModPhys.83.33} is given by
\begin{equation}\label{eq:k2}
    \overline{K}_n \approx \left(\frac{3}{2a}\right)^d \frac{1}{p} 
    = \left(\frac{3}{2a}\right)^{\log_2 n} \frac{1}{p} \equiv \overline{K}'_n.
\end{equation}
As opposed to our exact treatment below, this approximation can be applied only for the power-of-two case and, as we
will show, it is imprecise in relevant regimes of parameters $p$ and $a$. It is asymptotically precise only when both
$p$ and $a$ are small. In particular, for small $p$ and bigger values of $a$ (a very common regime, see e.g.
Ref.~\cite{quantum.2.93}) the approximation is much larger than the actual waiting times and hence the minimal repeater
performance in terms of lower bounds on the repeater rates is significantly underestimated. Moreover, the approximation
does not produce the correct scaling behavior. For the fully deterministic case $p=a=1$ we obviously have
$\overline{K}_n = 1$, because we immediately succeed in the very first attempt, but 
\begin{equation}
    \overline{K}'_n = (3/2)^{\log_2 n} = n^{\log_2 3/2} = n^{0.58\ldots} > \sqrt{n}.
\end{equation}
Thus, in this case the approximation is off by a factor of larger than $\sqrt{n}$. Neither the exact result
\eqref{eq:k1} nor the approximation \eqref{eq:k2} are directly applicable with a finite memory cutoff $m < +\infty$ or
with inclusion of classical communication times at higher nesting levels. The exact result \eqref{eq:Knm} for arbitrary
$m$ and $a=1$ is based on basic probability theory, hard to systematically generalize. How to obtain a systematic and
general framework and incorporate various effects such as arbitrary memory cutoffs and classical communication times
with our formalism is described in detail in later sections of our work. Here we shall now focus on general repeaters
for $m = +\infty$.

To find out the asymptotic behavior of $\overline{{K}_n}$ in the deterministic swapping case, $a=1$, and for $p\ll1$, we
expand $\left(1-p\right)^j$ with the Binomial Theorem and keep only the linear term in $p$ in each of the denominators
in the sum \eqref{eq:k1}. We have
\begin{equation}
    \overline{K}_n \approx \frac{1}{p} \sum_{j=1}^n (-1)^{j+1} \binom{n}{j} \frac{1}{j} = 
    \frac{1}{p}\sum_{j=1}^n \frac{1}{j} = \frac{H_n}{p},
\end{equation}
where $H_n = \sum_{j=1}^n 1/j$ is the $n$'th harmonic number. For large $n$ it can be approximated by 
\begin{equation}
    H_n = \gamma+\ln(n)+\frac{1}{2n} + \mathcal{O}\left(\frac{1}{n^2}\right),
\end{equation}
where $\gamma = 0.57721\ldots$ is the Euler–Mascheroni constant. Combining all these results, we get the following
simple approximation:
\begin{equation}
    \overline{K}_n \approx \frac{1}{p}\left(\gamma+\ln(n)+\frac{1}{2n}\right).
\end{equation}
For a fixed number of segments $n$ both approximations, this and the commonly used $\overline{K}'_n$, give the same
asymptotic growth of the commonly time as $\mathcal{O}(1/p)$, but $\overline{K}'_n$ gives a wrong factor. For a fixed
value of entanglement distribution success probability $p$, the waiting time really grows as $\mathcal{O}(\ln n)$, while
$\overline{K}'_n$ scales as $\mathcal{O}(n^{\log_2(3/2)}) = \mathcal{O}(n^{0.58\ldots})$, which grows faster than
$\mathcal{O}(\sqrt{n})$. This is a clear example where the usually employed approximations when used in the wrong
regime, i.e., $p \ll 1$ and $a=1$ as still relevant for practical quantum repeaters, lead to an inaccurate scaling.

The approximation $\overline{K}'_n$, given by Eq.~\eqref{eq:k2}, is applicable only in the power-of-two case, $n=2^d$,
and for deterministic swapping this approximation becomes rather inaccurate. In the next section, we present an approach
to compute the average waiting time exactly and for an arbitrary number of segments. This approach is based on the
venerable Markov chain theory, which has many applications in different branches of mathematics, physics, biology and
computer science.

\begin{figure*}
    \centering
    \begin{subfigure}[b]{0.49\textwidth}
        \includegraphics{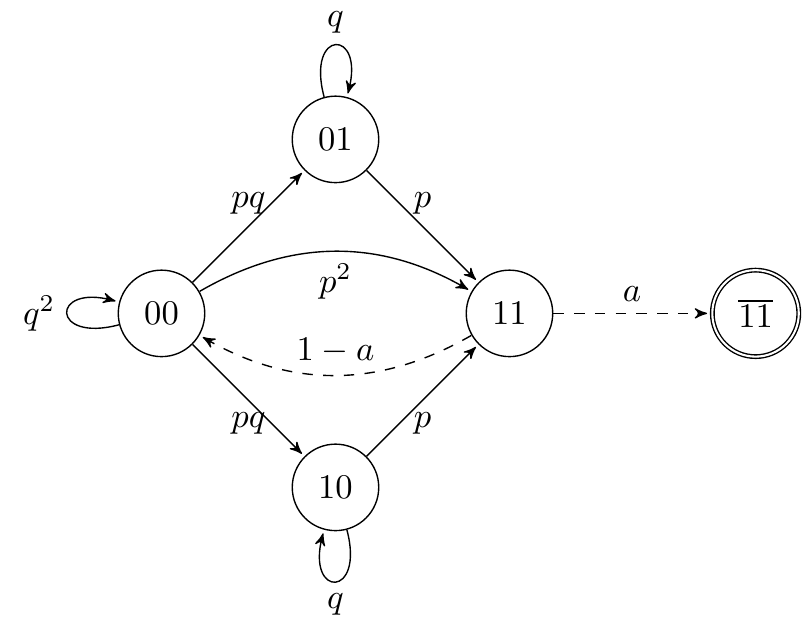}
        \caption{With zero-time transitions.}\label{fig:mc1z}
    \end{subfigure}
    \begin{subfigure}[b]{0.49\textwidth}
        \includegraphics{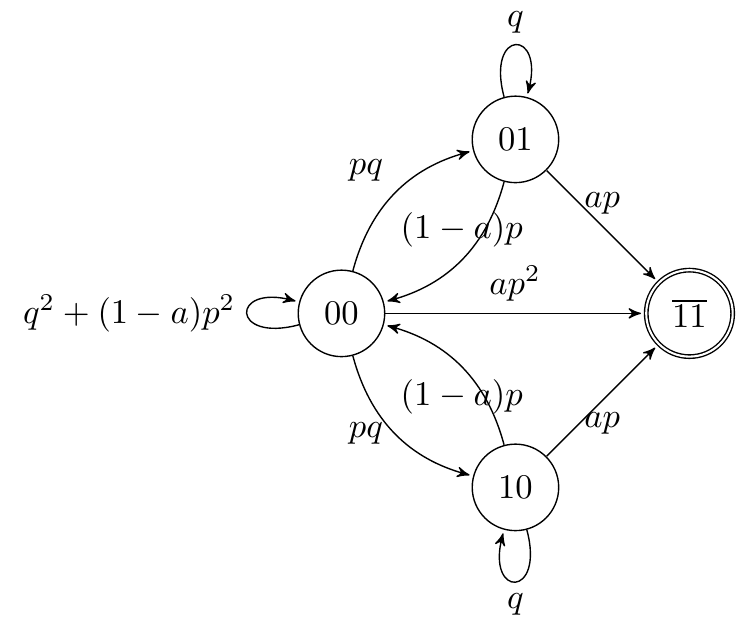}
        \caption{Without zero-time transitions.}\label{fig:mc1}
    \end{subfigure}
   \caption{Markov chain of a two-segment repeater.}\label{fig:ea}
\end{figure*}

\section{Markov chains}

Roughly speaking, a Markov chain is a formal description of a system that can be in several states and can go from one
state to another with known probability. We denote the state space of the system by $\mathcal{S} = \{s_1, \ldots,
s_N\}$. The transition probability from the state $s_i$ to the state $s_j$ is expressed as $p_{ij} = \mathbf{P}(s_i \to
s_j)$. The matrix $P = (p_{ij})^N_{i,j=1}$ of these probabilities is referred to as the transition probability matrix
(TPM) of the system. We can apply this formalism to study properties of the waiting time of quantum repeaters.

\subsection{Simple example}

To illustrate the description of quantum repeaters with the Markov chain approach, consider a two-segment repeater. It
can be only in five states, which we denote as $00$, $01$, $10$, $11$ and $\overline{11}$, where $0$ means that the
segment has no distributed entanglement yet and $1$ means that entanglement was successfully distributed. The overline
means that the corresponding segments have been successfully swapped and represent now a single, longer segment. The
transition probabilities between these states are easy to compute; the states with all possible transitions between them
are shown in Fig.~\ref{fig:mc1z}. For example, the probabilities $\mathbf{P}(11 \to 00) = 1-a$ and $\mathbf{P}(11 \to
\overline{11}) = a$ mean that from the state $11$, where both segments have successfully generated entanglement and are
ready to be swapped, we can either succeed with probability $a$ and move to the state $\overline{11}$, or we fail with
probability $1-a$ and move back to the initial state $00$. The TPM of the two-segment repeater thus reads as
\begin{equation}
    P = 
    \begin{pmatrix}
        q^2 & pq & pq & p^2 & 0 \\
        0 & q & 0 & p & 0 \\
        0 & 0 & q & p & 0 \\
        1-a & 0 & 0 & 0 & a \\
        0 & 0 & 0 & 0 & 1
    \end{pmatrix}.
\end{equation}
The main property of any TPM is that all its row sums are equal to one, and this matrix obviously satisfies this
property. Note that one state is special --- there is no arrow originating from the state $\overline{11}$. Once the
system entered this state, it will stay there. To be completely precise, we would need to show a loop for this state
with probability 1, but instead we distinguish it with double circle. Such states are referred to as absorbing. 

A typical approach to study the waiting time is to ignore the time it takes to perform the swapping, as well as the
classical communication time associated, in general, with any nesting level beyond the initial entanglement distribution
(the inclusion of such more general effects will be discussed in Sec.~VIII). Thus, we only count the initial
entanglement distribution time (including its classical communication). The waiting time, the variable $K_2$ in the
example, is thus the number of steps from the initial state $00$ to the absorbing state $\overline{11}$ without counting
the dashed zero-time transitions. This problem is closely related to the well-known problem for Markov chains:
absorption time. The absorption time is the number of steps it takes to get from some starting state to an absorbing
state (a general Markov chain can have more than one absorbing state). If we could describe the quantum repeater by a
Markov chain without zero-time transitions, then the waiting time would be exactly the absorption time, whose properties
can be obtained directly from the corresponding TPM. The Markov chain without zero-time transitions for a two-segment
quantum repeater is shown in Fig.~\ref{fig:mc1}. We just need to remove the state $11$ and recompute the transition
probabilities. We thus reduced the problem of studying the waiting time of this simplest repeater to the absorption time
problem of the Markov chain shown in Fig.~\ref{fig:mc1} with the following TPM:
\begin{equation}\label{eq:PTM2}
    P_2 = 
    \begin{pmatrix}
        q^2 + (1-a)p^2 & pq & pq & a p^2 \\
        (1-a)p & q & 0 & ap \\
        (1-a)p & 0 & q & ap \\
        0 & 0 & 0 & 1
    \end{pmatrix}.
\end{equation}
This reduction illustrates a general phenomenon: either we use a larger chain with zero-time transitions, whose
transition probabilities are easier to compute, or we use a more compact chain, but should spend more efforts to
determine the transition probabilities. Below we show that this reduction to the absorption time problem is possible for
a general $n$-segment quantum repeater, but before this we shall review the absorption time problem for general Markov
chains.

\subsection{Absorption time}

Consider a Markov chain with a single absorbing state, which we assume to be the last one, $s_N$. The general theory can
be applied to a chain with several absorbing states, but we need a simpler case when there is only one such state. In
this case the TPM can be partitioned as follows:
\begin{equation}\label{eq:PQu}
    P = 
    \begin{pmatrix}
        Q & \vec{u} \\
        \vec{0}^{\mathrm{T}} & 1
    \end{pmatrix},
\end{equation}
where $Q$ is the matrix of transition probabilities between non-absorbing states, $\vec{u}$ is the vector of transition
probabilities from non-absorbing states to the absorbing one and $\vec{0}$ is the zero vector. From the basic property
of TPMs it is easy to see that $\vec{u} = (I-Q)\vec{1}$, where $I$ is the identity matrix of the appropriate dimension,
$N-1$, and $\vec{1} = (1, \ldots, 1)^T$ is the vector of the same dimension with all components equal to $1$. For any
state $s_i$ except the absorbing one $s_N$ we introduce a random variable $K(i)$ whose value is the number of steps it
takes to get to the absorbing state from the state $s_i$ (for the absorbing state we would have $K(N) = 0$, so it is not
useful to introduce this variable). We combine all these variables into the vector 
\begin{equation}
    \vec{K} = (K(1), \ldots, K(N-1))^{\mathrm{T}}.
\end{equation}
We now show that the full PDF of these random variables (and thus their averages, standard deviations and all higher
moments) can be easily expressed in terms of the submatrix $Q$ of the TPM $P$.
\begin{thrm}\label{thrm:1}
    The probability distribution of $\vec{K}$ is given by the following simple expression:
    \begin{equation}\label{eq:PDF}
        \vec{p}_k = Q^{k-1}\vec{u}.
    \end{equation}
    The average waiting time and the second moment read as
    \begin{equation}\label{eq:K12}
    \begin{split}
        \overline{\vec{K}} &= R\vec{1}, \\ 
        \overline{\vec{K}^{\circ 2}} &= (2R - I)\overline{\vec{K}},
    \end{split}
    \end{equation}
    where $\vec{x}^{\circ 2}$ is the component-wise square (also known as Hadamard square) of the vector $\vec{x}$ and
    $R = (I-Q)^{-1}$ is the fundamental matrix of the chain. The variance is $\bm{\sigma}^2 = \overline{\vec{K}^{\circ
    2}} - \overline{\vec{K}}^{\circ 2}$.
\end{thrm}
\begin{proof}
For any non-absorbing state $s_i$ we have
\begin{equation}
    (\mathbf{p}_k)_i = \sum_{j_1, \ldots, j_{k-1}} Q_{ij_1}Q_{j_1j_2}\ldots Q_{j_{k-2}j_{k-1}}u_{j_{k-1}},
\end{equation}
where the summation is over all non-absorbing states since, if one of the intermediate states is the absorbing one, then
we succeed in less then $k$ steps. In the vector form we have exactly Eq.~\eqref{eq:PDF}. Now that we know the full PDF,
we can compute the average waiting time and its standard deviation. By the definition of the average value and
Eq.~\eqref{eq:PDF}, we have
\begin{equation}
    \overline{\vec{K}} = \sum^{+\infty}_{k=1} k \vec{p}_k = \sum^{+\infty}_{k=1} k Q^{k-1} \vec{u}.
\end{equation}
It has been proven in Ref.~\cite{mc} that the series $\sum^{+\infty}_{k=1}Q^{k-1}$ converges, so we have
\begin{equation}
\begin{split}
    \sum^{+\infty}_{k=1} k Q^{k-1} &= \left(\sum^{+\infty}_{k=1}Q^{k-1}t^k\right)'_{t=1} \\
    &= \left(t(I-Qt)^{-1}\right)'_{t=1} = (I-Q)^{-2},
\end{split}
\end{equation}
where the prime denotes the derivative with regards to $t$. From this it immediately follows that
\begin{equation}
    \overline{\vec{K}} = (I-Q)^{-2} \vec{u} = (I-Q)^{-2} (I-Q)\vec{1} = R\vec{1}.
\end{equation}
The second moment can be computed in a similar way:
\begin{equation}
    \overline{\vec{K}^{\circ 2}} = \sum^{+\infty}_{k=1} k^2 \vec{p}_k = \sum^{+\infty}_{k=1} k^2 Q^{k-1} \vec{u}.
\end{equation}
We have
\begin{equation}
\begin{split}
    \sum^{+\infty}_{k=1} k^2 Q^{k-1} &= \left(t\left(\sum^{+\infty}_{k=1}Q^{k-1}t^k\right)'\right)'_{t=1} \\
    &= (I+Q)(I-Q)^{-3},
\end{split}
\end{equation}
so we get
\begin{equation}
    \overline{\vec{K}^{\circ 2}} = (I+Q)(I-Q)^{-2}\vec{1} = (I+Q)(I-Q)^{-1}\overline{\vec{K}}.
\end{equation}
We just need to show that $(I+Q)(I-Q)^{-1} = 2(I-Q)^{-1}-I$, which is trivially verified by multiplying both sides with
$I-Q$. This concludes the proof.
\end{proof}
This result shows that the fundamental matrix $R=(I-Q)^{-1}$ is an important characteristic of the corresponding Markov
chain. If we can obtain this inverse in a meaningful form analytically, we can immediately compute the desired
characteristics of the waiting time. If it is infeasible to find the inverse analytically, we can go the numerical way.
The vector of average waiting times $\overline{\vec{K}}$ is the solution of the following system of linear equations:
\begin{equation}\label{eq:Ksle}
    (I-Q)\overline{\vec{K}} = \vec{1}.
\end{equation}
It is numerically more robust to solve the system \eqref{eq:Ksle} directly than to compute the matrix $R$ and multiply
it by $\vec{1}$. The second moment satisfies the system of linear equations
\begin{equation}
    (I-Q)\overline{\vec{K}^{\circ 2}} = (I+Q)\overline{\vec{K}}.
\end{equation}
Having found the first moment $\overline{\vec{K}}$, we can solve this system to obtain the second moment and the
standard deviation. Similarly, higher moments can be obtained. We thus can compute all statistical properties of each
random variable $K(i)$.

\subsection{Application to quantum repeaters}

Now we apply the general theory developed in the previous subsection to quantum repeaters and show that the problem of
determining the statistical properties of the waiting time can be reduced to the absorption time problem of an
appropriate Markov chain. In the previous subsection we demonstrated how to compute the absorption time of the system
started in any state, but in the applications to quantum repeaters we usually need only the absorption time of the
system started in the initial state (which we always assume to be the first), so we can just take the first component of
the vectors $\overline{\vec{K}}$ and $\bm{\sigma}$. Unfortunately, we cannot solve a system of linear equations just for
one variable without finding the values of the others, even if we later discard those other values. 

Each segment of an $n$-segment quantum repeater can be in two states, either entanglement has been distributed (which
happens with probability $p$) or it has not (with probability $q = 1-p$). Moreover, some groups of segments can be
successfully swapped, which we will denote by an overline over the swapped segments. The state of a quantum repeater can
be fully described by $n$-digit binary strings with overlines over all possible groups of 1s; individual segments are
shown without overlines. For a two-segment repeater there are five states, which were listed above. For a three-segment
repeater, there are 13 states:
\begin{alignat*}{7}
    &000 \quad &010 \quad &100 \quad &110 \quad &0\overline{11} \quad &\overline{11}0 \quad &\overline{111} \\
    &001 \quad &011 \quad &101 \quad &111 \quad &1\overline{11} \quad &\overline{11}1. \quad &
\end{alignat*}
The number of states grows exponentially with the number of segments. More precisely, the number of states is given by
the following
\begin{lmn}
    The number of states $N_n$ of an $n$-segment repeater is the $(2n+1)$-th Fibonacci number: $N_n = F_{2n+1}$.
\end{lmn}
\begin{proof}
The statement of the Lemma is correct for $n=2$ and $n=3$: in the former case there are $F_5 = 5$ states, $00$, $01$,
$10$, $11$ and $\overline{11}$, and in the latter case there are $F_7 = 13$ states as listed above. The odd-index
Fibonacci numbers $F_{2n+1}$ satisfy the following recurrence relation:
\begin{equation}
    F_{2n+5} = 3F_{2n+3} - F_{2n+1},
\end{equation}
which is easy to obtain from the defining relation $F_{n+2} = F_{n+1} + F_n$. We show that the numbers $N_n$ satisfy the
relation
\begin{equation}
    N_{n+1} = 3N_n - N_{n-1},
\end{equation}
which will prove that $N_n = F_{2n+1}$.

The set of all $(n+1)$-digit binary strings with overlines can be partitioned into two subsets $\mathcal{S}_1$ and
$\mathcal{S}_2$: the first subset contains those strings that do not end in an overline, and the second subset contains
those strings that do. Every string from $\mathcal{S}_1$ can be obtained by suffixing all possible states of an
$n$-segment repeater with zero and one, thus $|\mathcal{S}_1| = 2N_n$. It is easy to see that the strings of the set
$\mathcal{S}_2$ are in one-to-one correspondence with the states of an $n$-segment repeater that end with one: the
ending overline in an $n$-digit string $* \ldots \overline{*1}$ can be extended to $* \ldots \overline{*11}$ or a new
overline can be introduced if the last 1 is not overlined, $* \ldots *1 \to * \ldots *\overline{11}$. It follows that
$|\mathcal{S}_2|$ is the number of $n$-digit strings with overlines that end with one. This number can be obtained by
subtracting from the total number of strings $N_n$ the number of strings that end with zero, which are in one-to-one
correspondence with $(n-1)$-digit strings with overlines. We thus have the relation $|\mathcal{S}_2| = N_n - N_{n-1}$,
which gives $N_{n+1} = |\mathcal{S}_1| + |\mathcal{S}_2| = 3N_n - N_{n-1}$, so the proof is complete.
\end{proof}

From this Lemma it follows that the number of states grows exponentially as $N_n = F_{2n+1} = \mathcal{O}(\lambda^n)$,
with $\lambda = \varphi^2 = 2.61\ldots$, where $\varphi = (1+\sqrt{5})/2$ is the golden ratio constant. However, not all
of them are really needed to describe the process of entanglement distribution. The actual number of states necessary to
describe this process depends on the scheme used to perform swappings. We always start in the initial state $0 \ldots
0$, and not every state is reachable for a given swapping scheme. For example, for the doubling scheme, which is
exclusively studied in the literature, in the case of $n=4$ the state $0\overline{11}0$ is unreachable, since such a
swapping is forbidden in this scheme. On the other hand, for a scheme that we refer to as \textit{dynamical}, when we
swap everything that is ready, $0\overline{11}0$ is reachable, so both states $0110$ and $0\overline{11}0$ must be
included. Any scheme with a predetermined rule for the swappings, such as doubling, we refer to as \textit{fixed}.
Binary trees corresponding to fixed schemes for small quantum repeaters are shown in Fig.~\ref{fig:fix}.

\begin{figure}
    \includegraphics{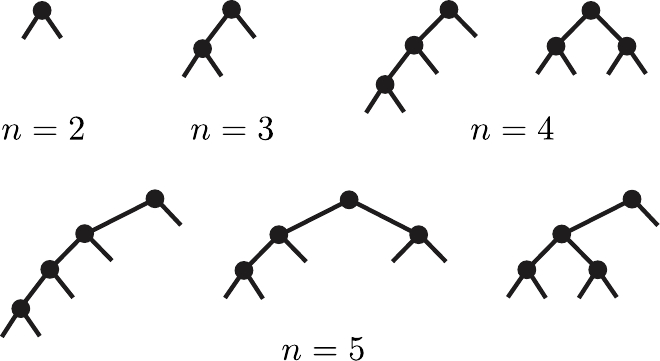}
    \caption{Fixed schemes for small quantum repeaters.}\label{fig:fix}
\end{figure}

Moreover, if we want to reduce the problem of waiting time computing to the absorption time problem, we also need to
remove all zero-time transitions and recompute the transition probabilities between the remaining states. This will
reduce the size of the corresponding Markov chain, but will also make the computation of transition probabilities
slightly more complicated by comparison with the Markov chain with zero-time transitions. Below we show that for the
doubling scheme, which is applicable only if the number of segments is a power of two, $n = 2^d$, the number of states
needed to implement the scheme without zero-time transitions is $N^{\mathrm{rec}}_n = 2^n = 2^{2^d}$. In this
``recursive'' scheme we divide the segments into two equal parts, wait when both are ready and then try to perform the
last swapping to distribute the entanglement over the whole repeater. We apply the same procedure to each half
recursively, which is always possible since the number of segments is a power of two. This is illustrated by
Fig.~\ref{fig:qr}c. Different schemes may require different subsets of the full set of states, but because the doubling
scheme is the only one that has been analyzed in the literature, we mainly use this scheme to illustrate our method.
Nonetheless, the method is completely general and can be applied to any swapping scheme. One of such more general
(non-doubling) schemes is illustrated by Fig.~\ref{fig:qr}d.

We illustrate the process of removing zero-state transitions and the recomputation of transition probabilities by the
case of $n = 3$ and a fixed scheme when we try to swap segments 1 and 2 first and then, having successfully swapped
them, try to swap the combined segments with segment 3. Since we ignore the time it takes to perform swapping and the
classical communication time needed to restart the process in the case of swapping failure, transitions between these
states take different times. For example, the transition $000 \to 000$ takes one time unit (and happens with probability
$q^3$), while the transition $011 \to 000$ happens instantaneously (with probability $1-a$). The states needed to
implement this scheme without zero-time transitions are
\begin{equation}\label{eq:--|-}
    000 \quad 001 \quad 010 \quad 011 \quad 100 \quad 101 \quad \overline{11}0 \quad \overline{111}.
\end{equation}
For example, the state $0\overline{11}$ is unreachable from the initial state $000$ and should not be included into this
set. The state $110$ has a zero-time transition to the state $\overline{11}0$ and should not be included as well. Note
that the set of states depends on the chosen scheme. If we consider a scheme where we swap segments 2 and 3 first and
then swap them with segment 1, the set of states reads as
\begin{equation}\label{eq:-|--}
    000 \quad 100 \quad 001 \quad 101 \quad 010 \quad 110 \quad 0\overline{11} \quad \overline{111}.
\end{equation}
Below we show how to determine this set of states for any fixed scheme with an arbitrary number of segments $n$, but now
we illustrate how to recompute the probabilities. Consider the transition $000 \to 000$. In the original Markov chain
with zero-time transitions the probability of this transition is $\mathbf{P}(000 \to 000) = q^3$. In the new, smaller
chain, we can stay in the state $000$ in three ways: (i) none of the entanglement distributions in the three segments
succeed; this happens with probability $q^3$, (ii) the segments 1 and 2 succeed, but swapping them fails; this happens
with probability $(1-a)p^2q$, (iii) all three segments succeed, swapping the first two also succeeds, but swapping them
with the last one fails; this happens with probability $a(1-a)p^3$. We thus see that the new transition probability is
the sum of these terms, $\mathbf{P}(000 \to 000) = q^3 + (1-a)p^2q + a(1-a)p^3$. This example illustrates what we meant
by saying that the new Markov chain will be smaller, but the transition probabilities will be more tricky to compute. In
the next subsection we present an algorithm to compute the TPM of the Markov chain corresponding to any fixed scheme.
Our approach is applicable to arbitrary schemes, not only to fixed ones, but the TPMs of dynamical schemes must be
constructed by other means. Ones the TPM is constructed, we can use Eqs.~\eqref{eq:PDF}-\eqref{eq:K12} to compute the
statistical characteristics of the chosen scheme. 

\subsection{Algorithm for TPM construction}

We now describe the algorithm for constructing the TPM $P$ of a quantum repeater constructed from two sub-repeaters. We
show how to express its TPM $P$ in terms of the sub-repeater TPMs $P'$ and $P^{\prime\prime}$. The schemes used for these
smaller repeaters can be arbitrary. They could also be fixed (and thus $P'$ and $P^{\prime\prime}$ could be obtained by
recursive applications of this algorithm), or they could be dynamical, or one fixed and the other dynamical --- any
possible combination will work. When both smaller repeaters successfully distribute entanglement over them, we try to
perform the last swapping and either succeed and distribute entanglement over the whole repeater or fail and have to
start this process from scratch. If the schemes used for smaller repeaters are also fixed all the way down, then the
whole scheme is fixed --- the order we perform swappings is fixed and does not depend on the order in which the segments
successfully distribute entanglement. But, as we have already noted, these schemes do not have to be fixed. The only
restriction of this construction is that the division of the segments is fixed on the highest level, and what happens
below can be arbitrary.

Let us assume that an $n$-segment repeater is divided into two parts with $n'$ and $n^{\prime\prime}$ segments, so that
$n = n' + n^{\prime\prime}$. The schemes of the smaller repeaters are described by the states $s'_i$, $i = 1, \ldots,
N'$, and $s^{\prime\prime}_j$, $j = 1, \ldots, N^{\prime\prime}$, respectively. As usual, we assume that $s'_1 = 0
\ldots 0$ ($n'$ zeros) and $s^{\prime\prime}_1 = 0 \ldots 0$ ($n^{\prime\prime}$ zeros) are the initial states of the
sub-repeaters and $s'_{N'} = \overline{1 \ldots 1}$, $s^{\prime\prime}_{N^{\prime\prime}} = \overline{1 \ldots 1}$ are
the absorbing ones. Because swapping between the segments of different sub-repeaters are forbidden, the scheme for the
whole repeater can be described by the states $s'_i s^{\prime\prime}_j$ (which means the concatenation of strings) with
one exception: the state $s'_{N'}s^{\prime\prime}_{N^{\prime\prime}} = \overline{1 \ldots 1} \, \overline{1 \ldots 1}$
must be replaced by the absorbing state $\overline{1 \ldots 1}$ ($n$ ones) for the whole repeater. We thus have that
this schemes requires $N = N' N^{\prime\prime}$ states. 

To illustrate this construction, consider a 3-segment repeater. We can divide it in only two ways: either like $--|-$ or
like $-|--$. The states of a 1-segment ``repeater'' are $0$, $1$, and the states of a 2-segment repeater (without
zero-time transitions) are $00$, $01$, $10$, $\overline{11}$. In the first case $s'_i$ denote the 1-segment repeater
states and $s^{\prime\prime}_j$ denote the 2-segment repeater states. It is easy to see that for such a division we
obtain exactly the states \eqref{eq:--|-} for the whole repeater. For the second division we get the states
\eqref{eq:-|--}.

We now need to compute transition probabilities between different states $s'_i s^{\prime\prime}_j$. As a general rule,
we have
\begin{equation}\label{eq:P1}
    \mathbf{P}(s'_i s^{\prime\prime}_j \to s'_k s^{\prime\prime}_l) 
    = \mathbf{P}'(s'_i \to s'_k) \mathbf{P}^{\prime\prime}(s^{\prime\prime}_j \to s^{\prime\prime}_l),
\end{equation}
provided that $s'_k s^{\prime\prime}_l$ is neither the initial nor the absorbing state of the whole repeater. For the
absorbing state we have the following simple rule:
\begin{equation}\label{eq:P2}
\begin{split}
    \mathbf{P}(s'_i s^{\prime\prime}_j &\to \overline{1 \ldots 1}) \\
    &= a \mathbf{P}'(s'_i \to \overline{1 \ldots 1}) \mathbf{P}^{\prime\prime}(s^{\prime\prime}_j \to \overline{1 \ldots 1}),
\end{split}
\end{equation}
because the overall success is determined by the success of the transitions inside the sub-repeaters and the success of
the swapping at the outermost level. For the initial state the rule is slightly more complicated and reads as
\begin{equation}\label{eq:P3}
\begin{split}
    \mathbf{P}(s'_i s^{\prime\prime}_j &\to 0 \ldots 0) \\
    &= \mathbf{P}'(s'_i \to 0 \ldots 0) \mathbf{P}^{\prime\prime}(s^{\prime\prime}_j \to 0 \ldots 0) \\
    &+ (1-a)\mathbf{P}'(s'_i \to \overline{1 \ldots 1}) \mathbf{P}^{\prime\prime}(s^{\prime\prime}_j \to \overline{1 \ldots 1}),
\end{split}
\end{equation}
since we can return to the initial state in two ways --- inside each sub-repeater individually or by going to the top
and failing to swap on the outermost level. All these relations, the general one given by Eq.~\eqref{eq:P1} and the two
exceptions, given by Eqs.~\eqref{eq:P2}-\eqref{eq:P3}, can be written in a compact matrix form. In the case of $a=1$ we
would have a very simple relation $P = P' \otimes P^{\prime\prime}$, but in the nondeterministic case this relation must
be modified. Namely, if we follow our standard convention that the initial state is the first in the TPM and the
absorbing state is the last, then $P$ is obtained from $P' \otimes P^{\prime\prime}$ by adding the last column,
multiplied by $1-a$, to the first column, and then multiplying the last column by $a$ (the elements of the last row are
untouched). In NumPy \footnote{\href{https://www.numpy.org/}{\texttt{https://www.numpy.org/}}} notation this lengthy explanation
can be compactly expressed as 
\begin{equation}\label{eq:Prec}
\begin{split}
	&\text{\texttt{$P$ = np.kron($P'$, $P^{\prime\prime}$)}} \\
    &\text{\texttt{$P$[:-1, 0] += (1-$a$)*$P$[:-1, -1]}} \\
    &\text{\texttt{$P$[:-1, -1] *= $a$}}
\end{split}
\end{equation}
This algorithm can be easily translated to other languages and computer algebra systems. We emphasize that in this
algorithm the sub-repeaters may have different numbers of segments, but the division of the repeater into sub-repeaters
is fixed --- we know beforehand when we will try the outermost swapping. The schemes used for the sub-repeaters are
arbitrary, not necessarily fixed. If we use the fixed schemes for the sub-repeaters, then we can recursively construct
the TPM of the whole repeater with this algorithm, which holds, for example, for the standard doubling scheme. For the
doubling scheme we start with the matrix $P_1$ of size $2$ of a single segment and double it $d$ times to get the TPM
$P_n$ of an $n$-segment repeater, where $n = 2^d$. Because doubling squares the size of the TPM, $P_n$ is of size
$2^{2^d} = 2^n$, as we claimed before.

\section{Examples}

We now demonstrate how to apply the algorithm of the previous section to small repeaters, where computing TPMs by hand
is still feasible. Moreover, here we also discuss dynamical schemes and compare their performance with that of fixed
schemes.

\subsection{One-segment repeater}

This is a completely trivial case, but we nevertheless include it to utilize it later as a building block of larger
repeaters. For a one-segment repeater the swapping success probability is not applicable, so in this case we have only
one parameter $p$. This repeater can be described by $F_3 = 2$ states $0$ and $1$. The TPM of such a repeater reads as
\begin{equation}\label{eq:PTM1}
    P_1 = 
    \begin{pmatrix}
        q & p \\
        0 & 1
    \end{pmatrix}.
\end{equation}
The $1 \times 1$ matrix $Q_1$ can be identified with a number $q$, so from the relations \eqref{eq:K12} we get 
\begin{equation}
    \overline{K}_1 = (1-q)^{-1}1 = \frac{1}{p}, \quad \overline{K^2_1} = \frac{1+q}{1-q}\overline{K}_1 = \frac{2-p}{p^2}.
\end{equation}
We thus have $\sigma^2_1 = \overline{K^2_1} - \overline{K}^2_1 = q/p^2$. From this relation we see that
$\sigma_1/\overline{K}_1 = \sqrt{q}$ and for small $p$ this ratio is close to one. It means that for small $p$ the
waiting time has a large deviation and cannot be precisely characterized by its average value alone. Of course, we could
easily obtain all these results directly from the geometric distribution of a single repeater segment, but our method is
applicable to larger repeaters too, where it gives the desired result much more easily than by computing the individual
probabilities.

\subsection{Deterministic swapping}

Here we consider an $n$-segment repeater with deterministic swapping. Because swappings cannot fail, their order does
not matter and all schemes are equivalent in this case. We need $2^n$ states to describe this scheme. We can identify
these states with the subsets of the set $[n] = \{1, \ldots, n\}$ of the segments (written as binary strings, we should
overline the consecutive runs of 1). For example, in the case of $n=3$ these states are 
\begin{equation}\label{eq:3det}
    000 \quad 001 \quad 010 \quad 0\overline{11} \quad 100 \quad 101 \quad \overline{11}0 \quad \overline{111}.
\end{equation}
The transition probability between the states associated with the subsets $I, J \subseteq [n]$ is given by
\begin{equation}\label{eq:pIJ}
    \mathbf{P}(I \to J) = 
    \begin{cases}
        0 & I \not \subseteq J \\
        p^{|J|-|I|} q^{n-|J|} & I \subseteq J
    \end{cases}.
\end{equation}
Nonzero transition probabilities are only from a set to any of its supersets since the set of ready segments can only
increase with each transition, but never decrease. We first need to check that such an assignment of probabilities is
correct, i.e., that the equality 
\begin{equation}
    \sum_{J \subseteq [n]} \mathbf{P}(I \to J) = 1
\end{equation}
is satisfied for all $I \subseteq [n]$. We have
\begin{equation}
    \begin{split}
      \sum_{J \subseteq [n]} \mathbf{P}(I \to J) &= \sum_{J \supseteq I} \mathbf{P}(I \to J) \\
      &= \sum^{n}_{j = |I|} \binom{n -|I|}{j - |I|} p^{j-|I|} q^{n - j} \\
      &= \sum^{n-|I|}_{j = 0} \binom{n -|I|}{j} p^{j} q^{n -|I| - j} = 1,
\end{split}
\end{equation}
since the last sum is just $(p + q)^{n-|I|} = 1$ by definition of $q = 1-p$. We thus have a correctly constructed Markov
chain without zero-transitions describing the entanglement distribution process in a quantum repeater with deterministic
swapping. To find the absorption time we need to solve the system of $2^n-1$ equations \eqref{eq:Ksle} for the vector
$\overline{\vec{K}} = \{\overline{K}(I)\}$, where the components are labeled by the strict subsets $I \subset [n]$. We
now show that the solution reads as
\begin{equation}\label{eq:KId}
    \overline{K}(I) = \sum^{n-|I|}_{j=1} (-1)^{j+1} \binom{n-|I|}{j} \frac{1}{1-q^j}.
\end{equation}
This expression for the empty set $I = \varnothing$ gives exactly the expression \eqref{eq:k1}. We need to check that
\begin{equation}
    \sum_{J \subseteq [n]} (\delta_{IJ} - \mathbf{P}(I \to J)) \overline{K}(J) = 1,
\end{equation}
which can be rewritten as
\begin{equation}\label{eq:PIJ}
    1 + \sum_{J \supseteq I} \mathbf{P}(I \to J) \overline{K}(J) = \overline{K}(I).
\end{equation}
The full proof is simple, but tedious so we just give two hints. The first thing to note is that the sum over all
supersets of $I$ can be replaced by the summation over a simple index $i$ from $0$ to $n - |I|$, introducing an
additional factor $\binom{n-|I|}{i}$ (there are so many supersets $J \supseteq I$ with $|J| = |I| + i$). The other hint
is that the double sum that appears after this transformation can be simplified with the following rule:
\begin{equation}
    \sum^{n-|I|}_{i=0} \,\, \sum^{n-|I|-i}_{j=1} = \sum^{n-|I|}_{j=1} \,\, \sum^{n-|I|-j}_{i=0}.
\end{equation}
The rest of the proof is just juggling with Binomial coefficients and applying the Binomial Expansion Theorem. We have
just reproduced the well-known result for the waiting time of a general repeater with deterministic swapping
\cite{PhysRevA.83.012323}. 

\subsection{Asymmetric case}

In this subsection, we generalize the result of the previous section and consider deterministic swapping with asymmetric
distribution probabilities. We no longer assume that all segments have the same entanglement distribution probability
$p$: the $i$-th segment has its own probability $p_i$. We use the notation
\begin{equation}
    p_I = \prod_{i \in I} p_i, \quad q_I = \prod_{i \in I} q_i
\end{equation}
for the product of probabilities over a set $I \subseteq [n]$. Note that in the case of all probabilities being equal,
$p_i = p$, we simply have $p_I = p^{|I|}$. We now derive an explicit expression for the average waiting time of such an
asymmetric repeater with deterministic swapping.

The transition probabilities read as
\begin{equation}
    \mathbf{P}(I \to J) = 
    \begin{cases}
        0 & I \not \subseteq J \\
        p_{J\setminus I} q_{\overline{J}} & I \subseteq J
    \end{cases},
\end{equation}
where $\overline{J} = [n]\setminus J$ is the complement of $J \subseteq [n]$. One can check that 
\begin{equation}
\begin{split}
    &\sum_{J \subseteq [n]} \mathbf{P}(I \to J) = \sum_{J \supseteq I} \mathbf{P}(I \to J) \\
    &= \sum_{J \supseteq I} p_{J \setminus I} q_{\overline{J}} = \prod_{i \in \overline{I}} (p_i + q_i) = 1,
\end{split}
\end{equation}
for all $I \subseteq [n]$, so these probabilities satisfy the Markov chain property. We now show that the explicit
solution of the system \eqref{eq:Ksle} in this case reads as
\begin{equation}\label{eq:KI}
    \overline{K}(I) = \sum_{\varnothing \subset J \subseteq \overline{I}} \frac{(-1)^{|J|+1}}{1-q_J}.
\end{equation}
We need to check that these functions satisfy the equalities \eqref{eq:PIJ}. Substituting Eq.~\eqref{eq:KI} into
Eq.~\eqref{eq:PIJ}, we will get a double sum, which can be transformed as follows:
\begin{equation}
    \sum_{J \supseteq I} \,\,\, \sum_{\varnothing \subset J' \subseteq \overline{J}} = 
    \sum_{\varnothing \subset J' \subseteq \overline{I}} \,\,\, \sum_{I \subseteq J \subseteq \overline{J'}}.
\end{equation}
We have
\begin{equation}
    \sum_{I \subseteq J \subseteq \overline{J'}} p_{J \setminus I} q_{\overline{J}} = 
    q_{J'} \prod_{i \in \overline{J'} \setminus I}(p_i + q_i) = q_{J'},
\end{equation}
so we get
\begin{equation}
\begin{split}
    1 &+ \sum_{J \supseteq I} \mathbf{P}(I \to J) \overline{K}(J) = 1 + \sum_{\varnothing \subset J' \subseteq \overline{I}}
    \frac{(-1)^{|J'|+1}}{1-q_{J'}} q_{J'} \\
    &= \sum_{J' \subseteq \overline{I}} (-1)^{|J'|} + \sum_{\varnothing \subset J' \subseteq \overline{I}}
    \frac{(-1)^{|J'|+1}}{1-q_{J'}} = \overline{K}(I).
\end{split}
\end{equation}
Note that if all $p_i = p$ are equal, the expression \eqref{eq:KI} reduces to Eq.~\eqref{eq:KId}, since for each $j = 1,
\ldots, n-|I|$ there are $\binom{n-|I|}{j}$ sets $J$ such that $\varnothing \subset J \subseteq \overline{I}$ and all
these sets correspond to the same $q_{J} = q^{|J|} = q^j$. In the following subsections we start analyzing the more
general and subtle case of non-deterministic swapping, $a<1$, for small repeaters.

\subsection{Two-segment repeater}

In the case of a two-segment repeater both parameters $p$ and $a$ are meaningful. This case has been considered before,
so here we just briefly discuss it. The TPM (without zero-time transitions) is given by Eq.~\eqref{eq:PTM2}. Note that
this matrix can be obtained by applying the algorithm \eqref{eq:Prec} to the matrix \eqref{eq:PTM1}. The fundamental
matrix $R_2 = (I - Q_2)^{-1}$ can be easily obtained and it becomes
\begin{equation}
    R_2 = 
    \frac{1}{ap(2-p)}
    \begin{pmatrix}
        1 & q & q \\
        1-a & a+q & (1-a)q \\
        1-a & (1-a)q & a + q
    \end{pmatrix}.
\end{equation}
For the vector $\overline{\vec{K}} = R_2 \vec{1}$ of average values, we have
\begin{equation}
    \overline{\vec{K}} = \frac{1}{ap(2-p)}
    \begin{pmatrix}
        1+2q \\
        1+2q-aq \\
        1+2q-aq
    \end{pmatrix}.
\end{equation}
It is easy to see that the first element of this vector, $\overline{K}_2$, coincides with the expression \eqref{eq:Kfm}
for $m = +\infty$, so we have just reproduced the well-known result for a two-segment repeater with non-deterministic
swapping. The variance now reads as
\begin{equation}\label{eq:sigma2}
    \sigma^2_2 = \overline{K}^2_2 - \frac{2p^3-3p^2-2p+4}{ap^2(2-p)^2}.
\end{equation}
For the ratio of the standard deviation to the average value in the typical case of small $p$, we have
\begin{equation}
    \lim_{p \to 0} \frac{\sigma_2}{\overline{K}_2} = \sqrt{1 - \frac{4}{9}a}.
\end{equation}
For small $a$ we again have $\sigma_2/\overline{K}_2 \approx 1$. In the more practical situation of large $a$ the
standard deviation is smaller (relative to the average value), but it never becomes smaller than
$(\sqrt{5}/3)\overline{K}_2 \approx 0.75\overline{K}_2$. So, even in this case the random variable $K_2$ cannot be
accurately characterized only by its average value.

\subsection{Three-segment repeater}

This is the first example of a not-power-of-two case. Here we have two fixed schemes (which are statistically equivalent
and represented by the $n=3$-scheme in Fig.~\ref{fig:fix}) and one dynamical scheme for performing swappings. The two
fixed schemes, segments 1 and 2 first, then segment 3 and segments 2 and 3 first, then segment 1 are equivalent and have
identical statistical properties. So, effectively, we have only one fixed scheme in this case. The TPM of this scheme
can be obtained from the matrices $P' = P_1$ and $P^{\prime\prime} = P_2$ according to the algorithm \eqref{eq:Prec},
corresponding to the second scheme (segments 2 and 3 first, then segment 1). The TPM $P_3$ and the corresponding
$\overline{K}_3$ become
\begin{widetext}
\begin{equation}\label{eq:PTM3}
\begin{split}
    P_3 &= 
    \begin{pmatrix}
        q^3 + (1-a)p^2 q + a(1-a)p^3 & pq^2 & pq^2 & ap^2q & pq^2 + (1-a)p^3 & p^2q & p^2q & a^2p^3 \\
        (1-a)pq + a(1-a)p^2 & q^2 & 0 & apq & (1-a)p^2 & pq & 0 & a^2p^2 \\
        (1-a)pq + a(1-a)p^2 & 0 & q^2 & apq & (1-a)p^2 & 0 & pq & a^2p^2 \\
        (1-a)p & 0 & 0 & q & 0 & 0 & 0 & ap \\
        a(1-a)p^2 & 0 & 0 & 0 & q^2 + (1-a)p^2 & pq & pq & a^2p^2 \\
        a(1-a)p & 0 & 0 & 0 & (1-a)p & q & 0 & a^2p \\
        a(1-a)p & 0 & 0 & 0 & (1-a)p & 0 & q & a^2p \\
        0 & 0 & 0 & 0 & 0 & 0 & 0 & 1
    \end{pmatrix}, \\
    \overline{K}_3 &= \frac{a^2(p^4 - 5p^3 + 10p^2 - 10p + 4) + a(2p^4 - 9p^3 + 17p^2 - 16p + 6) -4p^3 + 16p^2-23p+12}
                      {a^2p(2-p)(a(-p^3+3p^2-4p+2)+2p^2-5p+4)}.
\end{split}
\end{equation}
\end{widetext}
This time the expression for the fundamental matrix is not so compact as in the 2-segment case, but this is not a
problem for computer algebra systems. One can compute the standard deviation in the same way as we did it before,
but here we present only the limit of the ratio for small $p$:
\begin{equation}
    \lim_{p \to 0}\frac{\sigma_3}{\overline{K}_3} = \frac{\sqrt{48a^4 - 208a^3 + 68a^2+144a+144}}{4a^2+6a+12}.
\end{equation}
For small $a$ we again have $\sigma_3/\overline{K}_3 \approx 1$, while for large $a$ this ratio $\approx 7/11 \approx 0.64$.

We do not have to always follow a fixed way of performing swappings and may also try to swap any ready segments. In this
case we need a different subset of the full set of states to describe this scheme:
\begin{equation}\label{eq:3dyn}
    000 \quad 001 \quad 010 \quad 0\overline{11} \quad 100 \quad 101 \quad \overline{11}0 \quad \overline{111}.
\end{equation}
Note that these coincide with the states \eqref{eq:3det} for the case of a 3-segment repeater with deterministic
swapping. This is not surprising, since in the deterministic swapping case, all schemes are equivalent and thus can be
described by the same set of states. The TPM $P^{(\mathrm{dyn})}_3$ and the corresponding
$\overline{K}^{(\mathrm{dyn})}_3$ read as
\begin{widetext}
\begin{equation}\label{eq:PTM3dyn}
\begin{split}
    P^{(\mathrm{dyn})}_3 &= 
    \begin{pmatrix}
        q^3 + 2(1-a)p^2 q + a(1-a)p^3 & p q^2 + (1-a)p^3 & p q^2 & a p^2 q & p q^2 & p^2 q & a p^2 q & a^2 p^3 \\
        (1-a)pq + a(1-a)p^2 & q^2 + (1-a)p^2 & 0 & a p q & 0 & pq & 0 & a^2 p^2 \\
        2(1-a)pq + a(1-a)p^2 & (1-a)p^2 & q^2 & a p q & 0 & 0 & a p q & a^2 p^2 \\
        (1-a)p & 0 & 0 & q & 0 & 0 & 0 & a p \\
        (1-a)pq + a(1-a)p^2 & (1-a)p^2 & 0 & 0 & q^2 & pq & apq & a^2 p^2 \\
        a(1-a)p & (1-a)p & 0 & 0 & 0 & q & 0 & a^2 p \\
        (1-a)p & 0 & 0 & 0 & 0 & 0 & q & ap \\
        0 & 0 & 0 & 0 & 0 & 0 & 0 & 1
    \end{pmatrix}, \\
    \overline{K}^{(\mathrm{dyn})}_3 &= \frac{a^2(p^4 - 4p^3 + 6p^2 -5p+2) + a(2p^4 - 10p^3 + 21p^2 -22p+9)-4p^3+16p^2-22p+11}
    {a^2p(2-p)[a(-p^3+2p^2-2p+1)+3p^2-7p+5]}.
\end{split}
\end{equation}
\end{widetext}
In this case we cannot apply the algorithm \eqref{eq:Prec}, since this scheme is dynamical, and the TPM has been
computed directly from the definition of transition probabilities. We highlight the computation of some matrix elements.
For example, the transition $000 \to 000$ is possible in four ways: (i) no entanglement distribution is successful;
probability is $q^3$, (ii) the distribution in segments 1 and 2 was successful, but the swapping failed; probability is
$(1-a)p^2 q$, (iii) the distribution in segments 2 and 3 was successful, but the swapping failed; probability is
$(1-a)p^2 q$, (iv) the distribution in all three segments was successful, one swapping was also successful, but the
second swapping failed; probability is $a(1-a)p^3$. Summing the probabilities of all these exclusive possibilities, we
obtain the corresponding element of the TPM. In the case of a tie, when all segments successfully distribute
entanglement simultaneously, we always try to swap segments 1 and 2 first. Since the two possibilities are equivalent,
it does not matter whether we do this or first try to swap segments 2 and 3. We may even randomly choose between these
two options. The matrix is constructed under the assumption that segments 1 and 2 are always tried first, and that is
why the probabilities of the transitions $001 \to 001$ and $100 \to 100$ (and some others) differ. We can stay in the
state $001$ in two ways: (i) both segments 1 and 2 fail to distribute entanglement; probability is $q^2$, (ii) both
segments 1 and 2 succeed, but swapping them fails; probability is $(1-a)p^2$. The total probability is $q^2 + (1-a)p^2$.
On the other hand, we can stay in the state $100$ in only one way --- when segments 2 and 3 fail to distribute
entanglement, which happens with probability $q^2$. Had we chosen to swap segments 2 and 3 first, the corresponding
probabilities for these two transitions would have to be, well, swapped.

\begin{figure}
    \includegraphics{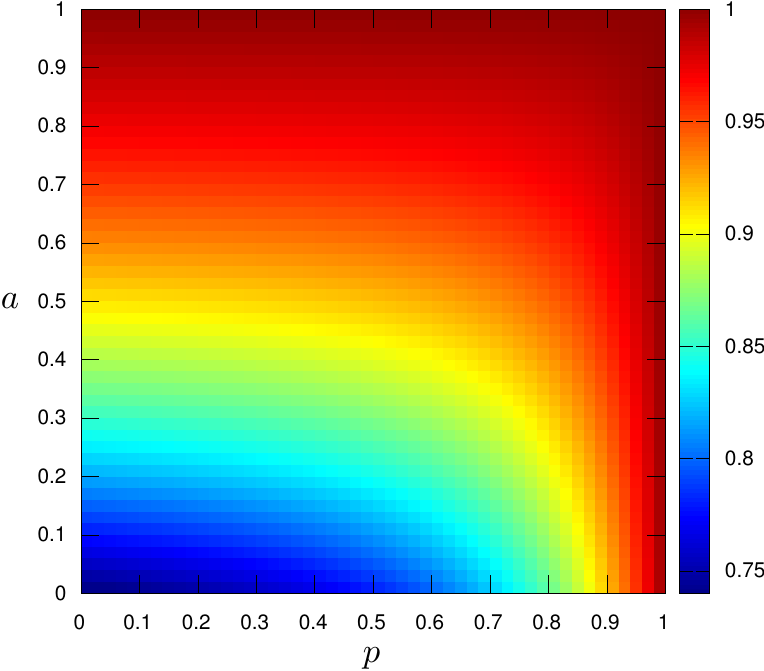}
    \caption{The ratio of the dynamical to fixed waiting times for a 3-segment repeater.}\label{fig:dyn2fix3}
\end{figure}

For the limit of the ratio of the standard deviation to the average value we now have
\begin{equation}
    \lim_{p \to 0}\frac{\sigma^{(\mathrm{dyn})}_3}{\overline{K}^{(\mathrm{dyn})}_3} 
    = \frac{\sqrt{12a^4 - 76a^3 - 59a^2+198a+121}}{2a^2+9a+11}.
\end{equation}
For small $a$ this ratio is close to 1, and for large $a$ it is $\approx 7/11 \approx 0.64$, as in the fixed-scheme
case. The ratio of the average value of the dynamical scheme to that of the fixed scheme is shown in
Fig.~\ref{fig:dyn2fix3}. We see that the dynamical scheme is slightly better, since we do not have to wait for concrete
segments and can start to swap them when they are ready. For $a = 1$ or for $p=1$ the two expressions for the average
value \eqref{eq:PTM3} and \eqref{eq:PTM3dyn} coincide (note that for $a=1$ both expressions coincide with \eqref{eq:k1}
for $n=3$), so the repeater rates are the same in this case. For small $p$ and $a$, we have
\begin{equation}
    \lim_{p, a \to 0} \frac{\overline{K}^{(\mathrm{dyn})}_3}{\overline{K}_3} = \frac{11}{15},
\end{equation}
so for small probabilities the dynamical scheme is approximately 25\% faster.

\subsection{Four-segment repeater}

The TPM $P_4$ of a 4-segment repeater with the recursive scheme can be easily obtained with the algorithm
\eqref{eq:Prec}. This is then effectively a fixed, doubling scheme (see Fig.~\ref{fig:fix} for $n=4$, right). For the
average waiting time of this scheme we have
\begin{widetext}
\begin{equation}
    \overline{K}_4 = \frac{2a^2p^4(p-1)(2p-3)-a(20p^5-72p^4+93p^3-53p^2+10p+4)+3(3-2p)^2(2p^2-3p+2)}
    {a^2p(2-p)(ap^2-(a+2)p+3)(-ap^3+4p^2-6p+4)}.
\end{equation}
\end{widetext}
For the ratio of the variance to the average value we have
\begin{equation}
    \lim_{p \to 0}\frac{\sigma_4}{\overline{K}_4} = \frac{\sqrt{-96a^3 - 272a^2-1728a+2916}}{54-4a}.
\end{equation}
Not surprisingly, for small $a$ this ratio is also close to 1. For large $a$, it is close to $41/125 = 0.328$. For four
segments there are more dynamical schemes than there are for three segments, but the waiting time of each of them can be
computed with our method. Another fixed scheme, which is not doubling, is illustrated by Fig.~\ref{fig:fix} for $n=4$,
left.

\section{Geometric approximation}

The PDF of the waiting time of a single segment is the classical geometric distribution, but for two and more segments
this is no longer exactly true. However, we now show that even in this case the PDF can be well approximated by a
geometric distribution with appropriately chosen parameters. This fact is easy to establish for the deterministic
swapping case. In fact, the probability that a single segment will finish in $k$ or less steps is equal to $1-q^k$
(because 
\begin{equation}
    1 - q^k = p \frac{1-q^k}{1-q} = \sum^{k-1}_{i=0} pq^i,
\end{equation}
or, more simply, the probability that a segment is not ready in $k$ steps is $q^k$), so the probability that $n$
segments will finish in $k$ or less steps is $(1-q^k)^n$. The probability $p_k = \mathbf{P}(K_n = k)$ that the $n$
segments will finish in exactly $k$ steps is then equal to 
\begin{equation}
    p_k = (1-q^k)^n - (1 - q^{k-1})^n,
\end{equation}
for all $k \geqslant 1$. Expanding the sums and making a simple transformation of the terms, we obtain
\begin{equation}
    p_k = q^{k-1}\sum^n_{j=1} (-1)^{j+1} \binom{n}{j} q^{(k-1)(j-1)}(1-q^j).
\end{equation}
In the limit $k \to +\infty$ the sum on the right tends to $n(1-q) = np$, so we have
\begin{equation}\label{eq:geomapprox}
    p_k \sim np q^{k-1},
\end{equation}
so for large $k$ the probability $p_k$ can be approximated by the geometric distribution with ratio $q$.

Non-deterministic swapping requires more sophisticated tools. A matrix $A$ is now called nonnegative (positive) if all
its elements are nonnegative (positive). A nonnegative matrix $A$ is called stochastic if all its row sums are equal to
one. If at least one of these inequalities is strict, then $A$ is called substochastic. A square matrix $A$ is called
primitive if some of its power $A^n$ are positive. The Perron-Frobenius Theorem \cite{pf} states that for any
nonnegative primitive matrix $A$ there is a real eigenvalue $\lambda_1$ with multiplicity one such that all other
eigenvalues $\lambda_j$ (which can be complex) are strictly smaller by absolute value, $|\lambda_j| < \lambda_1$ for $j
\geqslant 2$. If, in addition, $A$ is stochastic, then $\lambda_1 = 1$. If $A$ is substochastic, then $\lambda_1 < 1$.
It is easy to show that in the case of $a < 1$ (non-deterministic swapping) the matrix $Q$ of any scheme is primitive
and substochastic and thus is subject to the Perron-Frobenius Theorem. Substochasticity is obvious, so we need only to
prove that it is also primitive.

We show that all elements of $Q^2$ are strictly positive. In fact, let $\{s_1, \ldots, s_N\}$ be the states (a subset of
the full set of $F_{2n+1}$ states) that implement the given scheme without zero-time transitions. As usually, we assume
that the state $s_1 = 0 \ldots 0$ is initial. Then we have $(Q^2)_{ij} = \mathbf{P}(s_i \to s_1) \mathbf{P}(s_1 \to s_j)
+ \ldots$, where dots stand for other terms, which are nonnegative. Since $i$ is not the absorbing state we have
$\mathbf{P}(s_i \to s_1) > 0$, because there is always a chance that the last swapping fails, be it a fixed scheme or a
dynamical one (it is at this point that we need the assumption $a<1$), and from the initial state we can go to any other
state with nonzero probability, so we also have $\mathbf{P}(s_1 \to s_j) > 0$. It follows that $(Q^2)_{ij} > 0$. 

\begin{figure*}
    \includegraphics{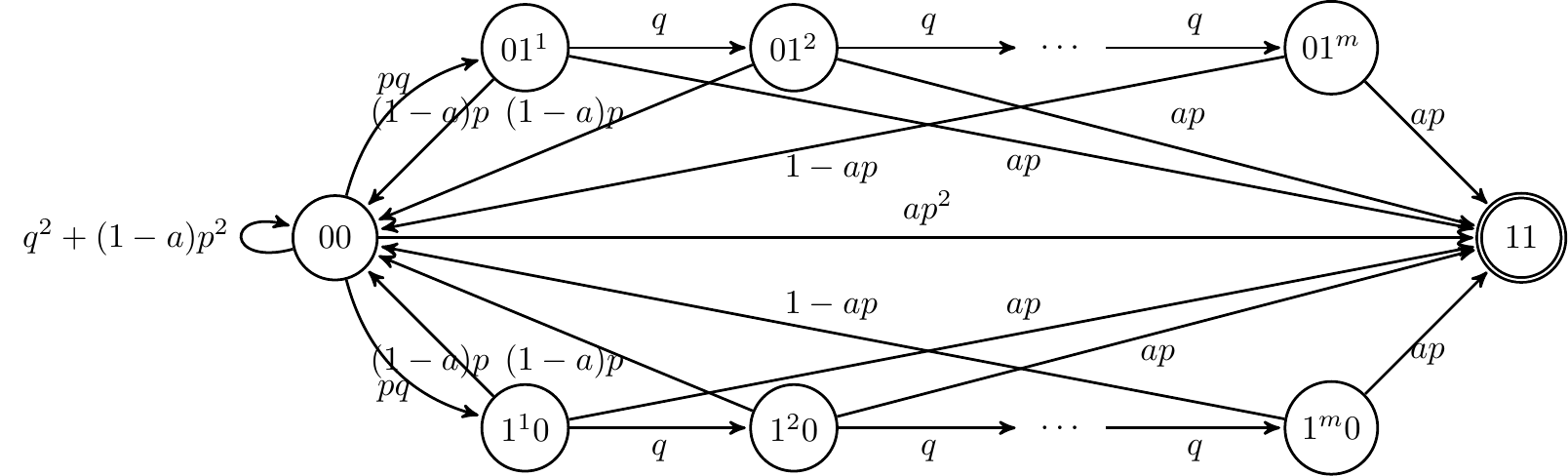}
    \caption{Markov chain of a two-segment repeater with finite memory cutoff and arbitrary swapping probability.}\label{fig:Kfm}
\end{figure*}

Applying the Perron-Frobenius Theorem, we can order the eigenvalues of $Q$ as $\lambda_1 > |\lambda_2| \geqslant \ldots
\geqslant |\lambda_{N-1}|$, and let $\vec{v}_1, \ldots, \vec{v}_{N-1}$ be the corresponding linearly independent
eigenvectors. We can express the vector $\vec{u}$ from the decomposition \eqref{eq:PQu} as a linear combination of these
eigenvectors, 
\begin{equation}
    \vec{u} = c'_1 \vec{v}_1 + \ldots + c'_{N-1} \vec{v}_{N-1}. 
\end{equation}
We just have to solve the linear system $V\vec{c}' = \vec{u}$, where $V$ is the matrix whose columns are the
eigenvectors $\vec{v}_j$, which is nondegenerate due to the linear independence of $\vec{v}_j$. We then have
\begin{equation}
    \vec{p}_k = Q^{k-1} \vec{u} = c'_1 \lambda^{k-1}_1 \vec{v}_1 + \ldots + c'_{N-1} \lambda^{k-1}_{N-1} \vec{v}_{N-1}.
\end{equation}
We need only the first element $p_k$ of this vector, so we finally obtain 
\begin{equation}
    p_k = c_1 \lambda^{k-1}_1 + \ldots + c_{N-1}\lambda^{k-1}_{N-1}, 
\end{equation}
where $c_j = c'_j v^0_j$ and $v^0_j$ is the first element of $\vec{v}_j$. For large $k$, we clearly have 
\begin{equation}
    p_k \approx c_1 \lambda^{k-1}_1,
\end{equation}    
because the largest eigenvalue $\lambda_1$ and its corresponding eigenvector are unique. Compare this expression with
Eq.~\eqref{eq:geomapprox}, where it is possible to find $\lambda_1 = q$ and $c_1 = np$ explicitly. In both cases,
starting from some sufficiently large $k$ the probabilities $p_k$ are well described by a geometric distribution,
similar to those for a single repeater segment. To our knowledge, this has never been shown explicitly. In fact, the
common wisdom seems to be that not only a single segment, but an entire quantum repeater follows exactly a geometric
distribution. We showed that this holds only approximately.

\section{Finite memory}

Let us first show that the Markov chain method can reproduce Eq.~\eqref{eq:Kfm}, which gives the waiting time of a
two-segment repeater with non-deterministic swapping and arbitrary memory cutoff. We introduce auxiliary states $01^i$
and $1^i0$, $i = 1, \ldots, m$. The superscript denotes the time the segment waits in the ready state. When this time
exceeds the limit, $m$ time units, the segment is forcefully reset to the initial state. The corresponding Markov chain
is shown in Fig.~\ref{fig:Kfm}. It is easy to construct the TPM corresponding to this chain and, with the help of a
computer algebra system, to verify that this TPM leads to Eq.~\eqref{eq:Kfm} for the average waiting time.

In the following, we show how to apply the Markov chain approach to repeaters with deterministic swapping, $a=1$, and
finite memory cutoff, $m$. In this case, we need another full set of states describing such a quantum repeater and
Markov chains with several absorbing states. The states are the tuples $(i_1, \ldots, i_n)$ with $i_j = 0, \ldots, m$
such that there is at least one $j = 1, \ldots, n$ with $i_j = 0$ or $i_j = 1$. Totally, there are $N = (m+1)^n -
(m-1)^n$ such states. Each component $i_j$ denotes the time passed since the $j$-th segment has successfully distributed
entanglement. The condition that there must be a component with value $0$ or $1$ means that a valid state is one that is
either not ready (some component is $0$) or has just become ready (no zero components but at least one is equal to $1$).
Tuples with all components larger than $1$ are not needed --- they describe those states where all the segments are in
the ready state for longer than necessary. Tuples with at least one zero component are non-absorbing states, and all the
others are absorbing. There are $N_0 = (m+1)^n - m^n$ non-absorbing states and $N_1 = m^n - (m-1)^n$ absorbing ones. If
we follow our usual convention and put all absorbing states at the end, then the PTM of this Markov chain can be
decomposed as follows:
\begin{equation}
    P = 
    \begin{pmatrix}
        Q & U \\
        0 & I
    \end{pmatrix},
\end{equation}
where $Q$ is an $N_0 \times N_0$ matrix (whose elements are probability transitions between non-absorbing states), $U$
is an $N_0 \times N_1$ matrix, $0$ is the zero $N_1 \times N_0$ matrix and $I = I_{N_1}$ is the $N_1 \times N_1$
identity matrix. As before, $R = (I - Q)^{-1}$ is the fundamental matrix of the chain (where $I = I_{N_0}$ is the $N_0
\times N_0$ identity matrix). According to the general result, Ref.~\cite{mc}, the absorption time as usual reads as
$R\vec{1}$ (so it does not matter how many absorbing states there are). We only need to compute the TPM $P$.

If fact, we only need to compute the matrix $Q$. Consider a non-absorbing state $(i_1, \ldots, i_n)$. Denote by $n_z =
|\{j|i_j = 0\}|$ the number of its zero components. Since the state is non-absorbing, $n_z > 0$. If for some $j$ we have
$i_j = m$, i.e., at least one segment has reached the maximal allowed time, then such a state can go only into two other
states: to the absorbing state $(i'_1, \ldots, i'_n)$ (where $i'_j = i_j$ if $i_j > 0$ and $i'_j = 1$ if $i_j = 0$) with
probability $p^{n_z}$, or to the initial state $(0, \ldots, 0)$ with probability $1-p^{n_z}$. If none of the segments
have reached the maximal time, then
\begin{equation}
\begin{split}
    \mathbf{P}((i_1, \ldots, i_n) &\to (i'_1, \ldots, i'_n)) = \\
    &\mathbf{P}(i_1 \to i'_1) \ldots \mathbf{P}(i_n \to i'_n),
\end{split}
\end{equation}
where $\mathbf{P}(0 \to 0) = 1-p$, $\mathbf{P}(0 \to 1) = p$, $\mathbf{P}(i \to i+1) = 1$, and all other probabilities
are zero. It is straightforward to construct this TPM programmatically in a computer algebra system for given parameters
$p$, $n$ and $m$, where, recall, we have set $a=1$. With the help of such a system, it is easy to verify that the
average waiting time produced by this TPM gives exactly the expression \eqref{eq:Knm}, $\overline{K}^{(m)}_n$, i.e., the
average waiting time for deterministic swapping and finite memory cutoff. Our formalism can also be used to treat the
most general case of arbitrary $a$ and $m$.

\begin{figure}[ht]
    \includegraphics{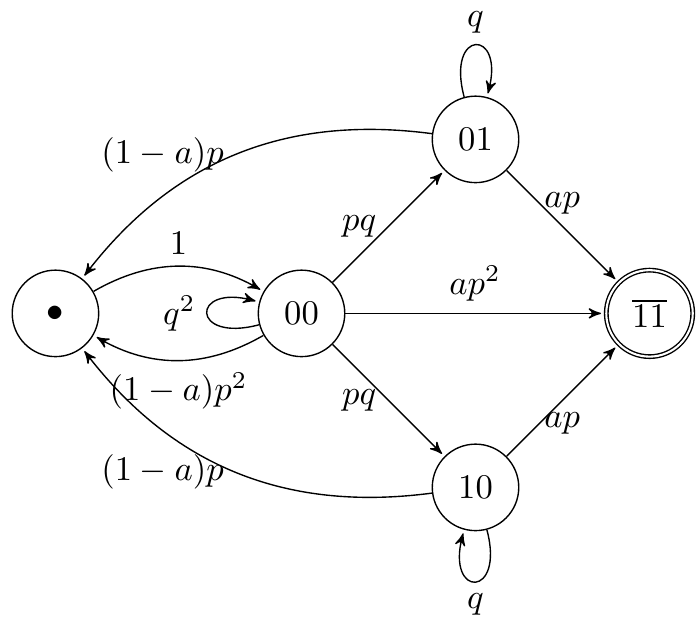}
    \caption{Markov chain of a two-segment repeater with classical communication.}\label{fig:2cc}
\end{figure}

\section{Classical communication}

In this section we add classical communication to the basic scheme considered before. Note that classical communication
times in a general quantum repeater with probabilistic swapping have been considered before, e.g. in
Ref.~\cite{PhysRevA.87.062335}, however, once again, those existing schemes are approximate depending on the usual
assumptions, as discussed in Sec.~II. We assume that the number of segments is a power of two, $n = 2^d$, and on each
level $i = 2, \ldots d$ it typically takes $2^{i-1}$ units of time to reinitiate the segments in the case of swapping
failure (because $2^i/2 = 2^{i-1}$). We also need classical communication to communicate success of a swapping
operation, but for simplicity we only consider the case of failure.  The idea is to introduce additional states that
correspond to these additional $2^{i-1}$ time units. In the case of $d=1$ we need one additional state, which we denote
by dot. The corresponding Markov chain of a two-segment repeater with classical communication is shown in
Fig.~\ref{fig:2cc}. Each transition in this figure takes exactly one unit of time. The corresponding TPM reads as
\begin{equation}\label{eq:P2cc}
    P = 
    \begin{pmatrix}
        q^2 & pq & pq & (1-a)p^2 & ap^2 \\
        0 & q & 0 & (1-a)p & ap \\
        0 & 0 & q & (1-a)p & ap \\
        1 & 0 & 0 & 0 & 0 \\
        0 & 0 & 0 & 0 & 1
    \end{pmatrix}.
\end{equation}
We now show how to construct TPMs for larger repeaters.

We describe an algorithm to ``double'' a given repeater with the TPM $P$ and the classical communication time of $c$
time units. The original repeater can employ any kind of scheme. Applying our algorithm recursively all the way down, we can
compute the TPM of a repeater with $n = 2^d$ segments and classical communication at all levels.

\begin{figure*}
    \centering
    \begin{subfigure}[b]{0.47\textwidth}
        \includegraphics{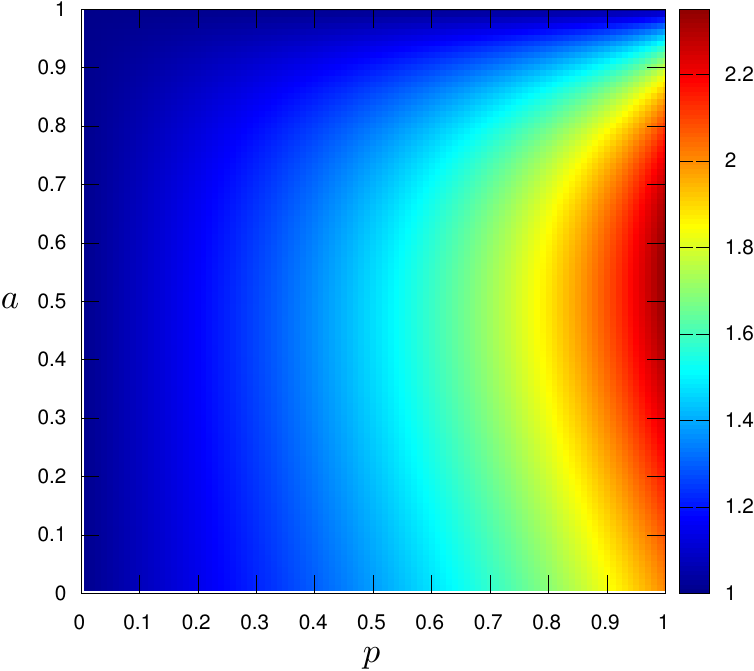}
    \end{subfigure}
    \begin{subfigure}[b]{0.47\textwidth}
        \includegraphics{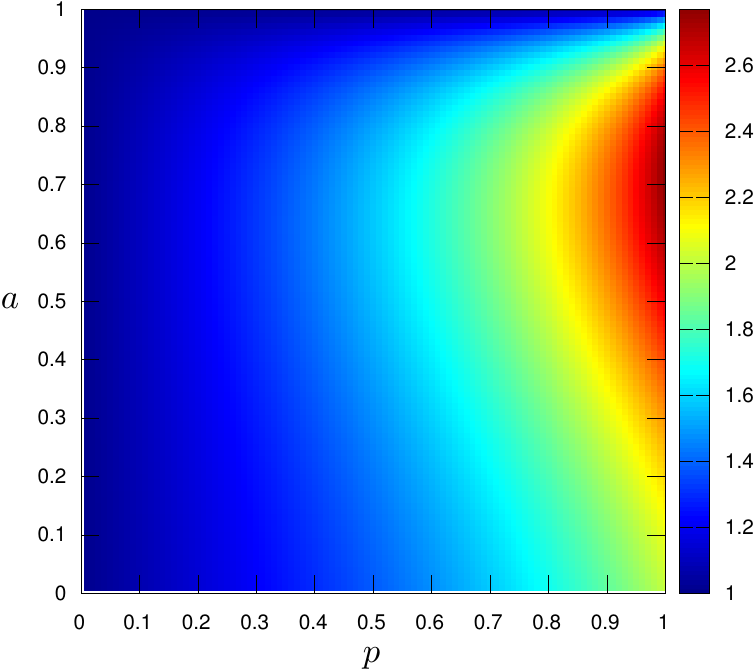}
    \end{subfigure}
   \caption{The ratio of the waiting time with classical communication to the waiting time without 
   for $n=8$ (left) and $n=16$ (right) segments.}\label{fig:cc}
\end{figure*}

Given a quantum repeater with the TPM $P$, we want to construct a quantum repeater of double length, performing the last
swapping exactly in the middle of the new repeater, and assuming that restarting the process in the case of failure
takes $c$ time units. If the original repeater is described by the states $s_1, \ldots, s_N$, the doubled repeater
without classical communication is described by the $N^2$ states $s_i s_j$, $i, j = 1, \ldots, N$. To include the
classical communication into the scheme we need $c$ additional states, which we denote as $\bullet_1$, \ldots,
$\bullet_c$, so the full set of states contains $N^2+c$ elements. We arrange these states in the following order: $s_i
s_j$ for all $i$, $j$ except the state $s_N s_N$, then $\bullet_i$, $i = 1, \ldots, c$, and finally $s_N s_N$. The
``doubled'' TPM $P_d$ can then be partitioned into nine blocks as
\begin{equation}
    P_d = 
    \begin{pmatrix}
        P_{ss} & P_{s\bullet} & P_{sN} \\
        P_{\bullet s} & P_{\bullet\bullet} & P_{\bullet N} \\
        0 & 0 & 1
    \end{pmatrix}.
\end{equation}
We give explicit expressions the first six blocks. Let partition $P \otimes P$ be as follows:
\begin{equation}
    P \otimes P = 
    \begin{pmatrix}
        Q_d & \vec{u}_d \\
        0 & 1
    \end{pmatrix}.
\end{equation}
Then we simply have $P_{ss} = Q_d$ since the transition probability between non-absorbing states $s_{i_1} s_{j_1}$ and
$s_{i_2} s_{j_2}$ is the product of transition probabilities $s_{i_1} \to s_{i_2}$ and $s_{j_1} \to s_{j_2}$. For the
next block we have
\begin{equation}
    P_{s\bullet} = 
    \begin{pmatrix}[c|c]
        & \ldots \\
        (1-a)\vec{u}_d & 0 \\
        & \ldots
    \end{pmatrix},
\end{equation}
because when the last swapping fails, we have to go to the state $\bullet_1$ and then restart the process through the
sequence of transitions $\bullet_1 \to \ldots \to \bullet_c \to s_1 s_1$. Each of these transitions happens with
probability one, so we also have
\begin{equation}
    P_{\bullet s} = 
    \begin{pmatrix}
        0 & 0 & \ldots & 0 & 0 \\
        0 & 0 & \ldots & 0 & 0 \\
        \hdotsfor{5} \\
        0 & 0 & \ldots & 0 & 0 \\
        1 & 0 & \ldots & 0 & 0
    \end{pmatrix}, \ 
    P_{\bullet\bullet} = 
    \begin{pmatrix}
        0 & 1 & \ldots & 0 & 0 \\
        0 & 0 & \ldots & 0 & 0 \\
        \hdotsfor{5} \\
        0 & 0 & \ldots & 0 & 1 \\
        0 & 0 & \ldots & 0 & 0 
    \end{pmatrix}.
\end{equation}
The other blocks are just $P_{sN} = a \vec{u}_d$ and $P_{\bullet N} = \vec{0}$. This gives us the full TPM of the
doubled repeater with classical communication. One can easily check that the matrix \eqref{eq:P2cc} can be obtained from
the matrix \eqref{eq:PTM1} with this algorithm with $c=1$. Fig.~\ref{fig:cc} illustrates the effect of classical
communication as the ratio of the averaged waiting time with classical communication to that without. One can see that
the maximal effect of classical communication grows with the number of segments, and also, as one would expect, for
larger probabilities $p$ and $a$ where the absolute waiting times decrease and so the net effects of classical
communication become more pronounced. For small $p$ and $a$, waiting times are so long anyway that the addition of
classical communication times hardly makes a difference. This is why in the regime of the commonly used approximation,
the classical communication times are typically neglected on higher levels. However, for a practical quantum repeater
with moderate $p$ and significantly larger $a$, for many segments, the classical communication should be included. In
the case $a=1$, it is already implicitly included via the elementary time unit for the initial entanglement
distribution. Note that once memory cutoffs and even explicit memory errors (such as dephasing) are included, classical
communication times become relevant even for small repeaters. 

\section{Generating functions}

An alternative approach to obtain the average waiting time is to use the method of generating functions \footnote{In an
earlier version of our work, we had already included this alternative approach. Later Vinay and Kok
\cite{PhysRevA.99.042313} put a greater emphasis on such methods and also demonstrated the usefulness of (probability)
generating functions for analyzing the rates and the statistics of a quantum repeater.}. In general, this method is
inferior to the Markov chain method, since numerically it is much easier to deal with matrices and vectors than with
functions. On the other hand, generating functions can be useful in more complicated situations, where the Markov chain
approach cannot be easily applied. We shall give an example for this in this section.

If $X$ is a random variable taking only non-negative integer values, then its probability generating function (PGF)
$g(t)$ is defined by 
\begin{equation}
    g(t) = \sum^{+\infty}_{k = 0} \mathbf{P}(X=k)t^k.
\end{equation}
Any PGF $g(t)$ satisfies the equality $g(1) = 1$. The PGF $g(t)$ contains full statistical information about the
corresponding random variable $X$, for example, the average value $\overline{X}$ and the variance $\sigma^2 =
\overline{X^2} - \overline{X}^2$ via $\overline{X} = g'(1)$ and $\sigma^2 = g^{\prime\prime}(1) + g'(1) - g'(1)^2$. We
now show how to reproduce Theorem~\ref{thrm:1}. 

Consider a chain with $N$ states $s_i$, $i = 1, \ldots, N$, and single absorbing state $s_N$. For any state $i = 1,
\ldots, N$ we introduce an integer-valued random variable $K_i$, which is defined to be the number of steps it takes to
get to the absorbing state starting in the state $s_i$ (note that $K_N = 0$ since we are already in the absorbing
state). Let $g_i(t)$ denote the PGF of $K_i$ and $\vec{g}(t) = (g_1(t), \ldots, g_{N-1}(t))^{\mathrm{T}}$ be the vector
of these PGFs. They are connected by the following system of linear equations 
\begin{equation}
    g_i(t) = \sum^N_{j = 1} p_{ij} t_{ij} g_j(t), 
\end{equation}
where $t_{ij} = t$ if the transition $s_i \to s_j$ takes one time unit and $t_{ij} = 1$ if $s_i \to s_j$ is a zero-time
transition. It is clear that $g_N(t) = 1$. In the matrix notation it reads as 
\begin{equation}
    \vec{g}(t) = Q(t) \vec{g}(t) + \vec{u}(t), 
\end{equation}
where $Q(t)$ is the matrix $Q$ when each element is multiplied by $t$ except zero-time transitions elements, which are left
unchanged. The vector $\vec{u}(t)$ is obtained in a similar way from $\vec{u}$, where $Q$ and $\vec{u}$ are defined by
the decomposition \eqref{eq:PQu}. The solution of this system reads as 
\begin{equation}
    \vec{g}(t) = (E - Q(t))^{-1}\vec{u}(t). 
\end{equation}    
This expression can be applied to an arbitrary absorbing chain with a single absorbing state and zero-time transitions.
In the case of a Markov chain without zero-time transitions all elements are multiplied by $t$, so we have $Q(t) = t Q$,
$\vec{u}(t) = t \vec{u}$, and the solution becomes 
\begin{equation}\label{eq:g}
    \vec{g}(t) = t(E - tQ)^{-1}\vec{u}. 
\end{equation}
It is easy to see that 
\begin{equation}
    \vec{g}(1) = (E-Q)^{-1}\vec{u} = (E-Q)^{-1}(E-Q)\vec{1} = \vec{1}, 
\end{equation}    
so this solution satisfies the basic property of a PGF. The derivative $\vec{g}'(t)$ is easy to compute. We have
\begin{equation}
    (E-tQ)\vec{g}(t) = t(E-Q)\vec{1}.
\end{equation}
Differentiating this relation we get
\begin{equation}
    -Q\vec{g}(t) + (E - tQ)\vec{g}'(t) = (E-Q)\vec{1}.
\end{equation}
Substituting $t = 1$ and taking into account that $\vec{g}(1) = \vec{1}$ we obtain the relation 
\begin{equation}
    \overline{\vec{K}} = \vec{g}'(1) = (E - Q)^{-1}\vec{1} = R\vec{1}.
\end{equation}
The second moment can be computed as 
\begin{equation}
    \overline{\vec{K}^{\circ 2}} = \vec{g}^{\prime\prime}(1) + \vec{g}'(1) = (2R-I)\overline{\vec{K}},
\end{equation}
which coincides with Eq.~\eqref{eq:K12}.

We now compute the probabilities $\vec{p}_k$ using the generating functions approach. Let us denote $\tilde{\vec{p}}_{k}
= \mathbf{P}(\vec{K} > k)$ and introduce the corresponding generating function (GF)
\begin{equation}
    \tilde{g}(t) = \sum^{+\infty}_{k=0} \tilde{\vec{p}}_k t^k.
\end{equation}
It is easy to see that the standard PGF can be expressed in terms of the standard generating function as
\begin{equation}
    g(t) = 1 - (1-t) \tilde{g}(t).
\end{equation}
For the random variables under discussion we explicitly get this GF from Eq.~\eqref{eq:g} 
\begin{equation}
    \tilde{\vec{g}}(t) = (E-tQ)^{-1} \vec{1}.
\end{equation}
For the $k$-th derivative we have
\begin{equation}
    \tilde{\vec{g}}^{(k)}(t) = k! (E-tQ)^{-(k+1)} Q^k \vec{1},
\end{equation}
from which we immediately obtain the probabilities
\begin{equation}
    \tilde{\vec{p}}_k = \frac{\tilde{\vec{g}}^{(k)}(0)}{k!} = Q^k \vec{1}.
\end{equation}
The probabilities $\vec{p}_k$ can be computed as
\begin{equation}
    \vec{p}_k = \tilde{\vec{p}}_{k-1} - \tilde{\vec{p}}_k = Q^{k-1}(E-Q)\vec{1} = Q^{k-1}\vec{u},
\end{equation}
which is exactly the expression \eqref{eq:PDF}. 

As a concrete example, let us compute the PGF of a two-segment repeater. According to Eq.~\eqref{eq:PTM2}, the matrix
$Q$ and the vector $\vec{u}$ read as
\begin{equation}
\begin{split}
    Q &= 
    \begin{pmatrix}
        q^2 + (1-a)p^2 & p q & p q \\
        (1-a)p & q & 0 \\
        (1-a)p & 0 & q
    \end{pmatrix}, \\
    \vec{u} &= (a p^2, a p, a p)^{\mathrm{T}}.
\end{split}
\end{equation}
From Eq.~\eqref{eq:g} we obtain
\begin{equation}\label{eq:g2}
    g(t) = \frac{a p^2 t (1+qt)}{1-(2-3p+(2-a)p^2)t + q(1-2p+ap^2)t^2},
\end{equation}
where $g(t)$ is the first element of the vector $\vec{g}(t)$. It is easy to check that $g'(1) = \overline{K}_2$ and 
\begin{equation}
    g^{\prime\prime}(1) + g'(1) - g'(1)^2 = \sigma^2_2,
\end{equation}
where the variance $\sigma^2_2$ is given by Eq.~\eqref{eq:sigma2}.

\begin{figure}
    \includegraphics{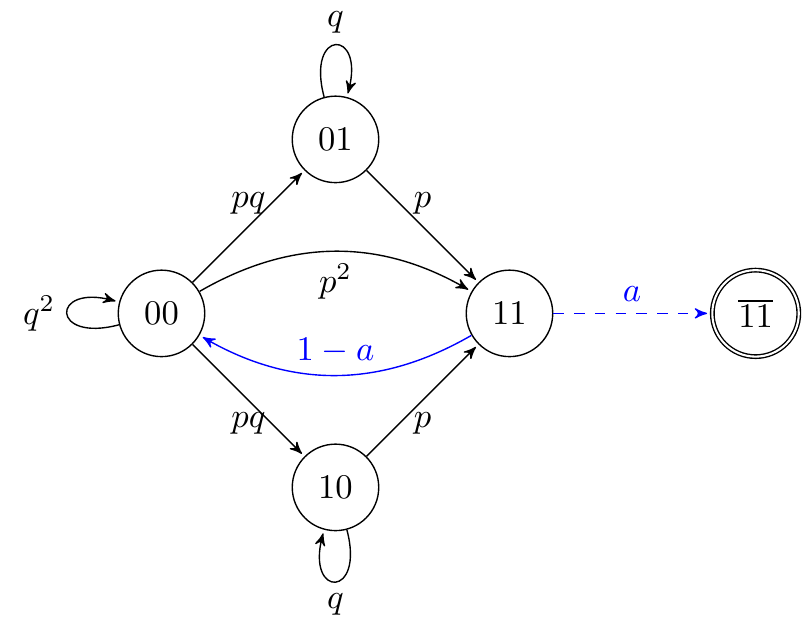}
    \caption{The Markov chain for a two-segment repeater with different time scales.}\label{fig:MCdt}
\end{figure}

Up to now we have always assumed that all time intervals can be expressed in terms of one common time unit. Let us
consider a more general problem, where restarting the process takes the time $\tau'$, which is independent of $\tau$
(while, for simplicity, here we again ignore the classical communication times when the swapping succeeds, though this
can be also easily included). Fig.~\ref{fig:MCdt} shows the Markov chain of such a repeater. All transitions marked by
black arrows take the time $\tau$. The transition marked by a blue arrow takes the time $\tau'$. The time for the
transition marked by the blue dashed arrow (swapping success) is ignored. The absorption time now reads as 
\begin{equation}
    T = K \tau + L \tau', 
\end{equation}
where the random variable $K$ is the number of black transitions and $L$ is the number of blue transitions. 

We mark different transitions by different variables ($t$, as before, marks attempts to distribute entanglement and $v$
marks unsuccessful attempts to perform swapping corresponding to the blue line in the figure). The matrix $Q$ with this
information included reads as
\begin{equation}
  Q(t, v) = 
  \begin{pmatrix}
    t(1-p)^2 & tp(1-p) & tp(1-p) & tp^2 \\
    0 & t(1-p) & 0 & tp \\
    0 & 0 & t(1-p) & tp \\
    v(1-a) & 0 & 0 & 0 
  \end{pmatrix}.
\end{equation}
The PGFs of the absorption times are given by 
\begin{equation}
\begin{split}
    \vec{g}(t, v) &= (E - Q(t, v))^{-1} \vec{u}(t, v), \\
    \vec{u}(t, v) &= (0, 0, 0, a)^{\mathrm{T}}.
\end{split}
\end{equation}
For the first component, $g = g(t, v)$, we have
\begin{widetext}
\begin{equation}
    g(t, v) = \frac{a p^2 t (1 + qt)}{1 - (2-3p+p^2 (1+(1-a)v))t + q (1-2p+p^2(1-(1-a)v))t^2}.
\end{equation}
\end{widetext}
If we set $v=1$ in this expression, i.e. if we ignore the additional information, we get the less detailed PGF given by
Eq.~\eqref{eq:g2}. Taking into account that 
\begin{equation}
    \overline{K} = g'_t(1, 1), \quad \overline{L} = g'_v(1, 1), 
\end{equation}    
we obtain 
\begin{equation}
    \overline{T} = \overline{K}_2 \tau + \frac{1-a}{a} \tau'.
\end{equation}
The variance $\sigma^2$ can be computed in the usual way, and the only new part is computing the correlation
\begin{equation}
    \overline{K L} = g^{\prime\prime}_{tv}(1, 1) = \frac{2(1-a)(3-2p)}{a^2 (2-p) p}. 
\end{equation}    
It is easy to see that for small $p$ and $a$ the main contribution to $\overline{T}$ is given by the first term,
$\overline{K}_2\tau$. 

\section{Lumpability}

The full set of states needed to describe an $n$-segment quantum repeater with infinite memory is $F_{2n+1} =
\mathcal{O}(2.61\ldots^n)$. The doubling scheme requires less states, namely $2^n$. Here we show that if the scheme has
some symmetry, this symmetry can be exploited to greatly reduce the number of states \footnote{In an earlier version of
our work, we had already included the lumbability method as a means to render the Markov-chain-based rate analysis more
efficient. Vinay and Kok \cite{PhysRevA.99.042313} employ and describe similar methods.}. We emphasize that the method
presented here is just an implementation trick and in no way influences the result produced by the original algorithm.
It just makes this algorithm much more practical than it otherwise would be.

The basic idea of the trick is the observation that in the absorption time problem we are interested only in the number
of steps from an initial state to the (or a) absorbing state. If, on the way to absorption, some intermediate states
give the same contribution to the number of states, it makes sense to combine them into one large state and work with
this more coarse Markov chain. This idea to replace the original Markov chain by a coarser one which still correctly
describes the desired property (the number of steps in this case) is formalized by the lumpability property. 

Let the states of the Markov chain be $s_1, \ldots, s_N$. Consider a partition of this set of states, 
\begin{equation}\label{eq:StatesP}
    \{s_1, \ldots, s_N\} = S_1 \cup \ldots \cup S_M,
\end{equation}
so that $S_i \cap S_j = \varnothing$ for $i \not= j$. If this partition satisfies the condition that the probability
\begin{equation}\label{eq:TPS}
    \mathbf{P}(S_i \to S_j) = \sum_{s_l \in S_j} \mathbf{P}(s_m \to s_l)
\end{equation}
does not depend on $s_m \in S_i$, then the chain is called lumpable with respect to the partition \eqref{eq:StatesP}.
This is illustrated by Fig.~\ref{fig:lump}, where the sums of solid and dashed lines must be equal and their common
value is considered to be the new probability transition, denoted by the bold line. We can construct a new Markov chain
with the states $S_1, \ldots, S_M$ and the transition probabilities between them are given by Eq.~\eqref{eq:TPS}. The
new, coarser Markov chain has fewer states and thus is easier to deal with, so if it preserves the information we need,
we can work directly with this new chain instead of the original, larger one. We now give two concrete examples of
lumpability of Markov chains in the context of quantum repeaters.

\begin{figure}
    \includegraphics{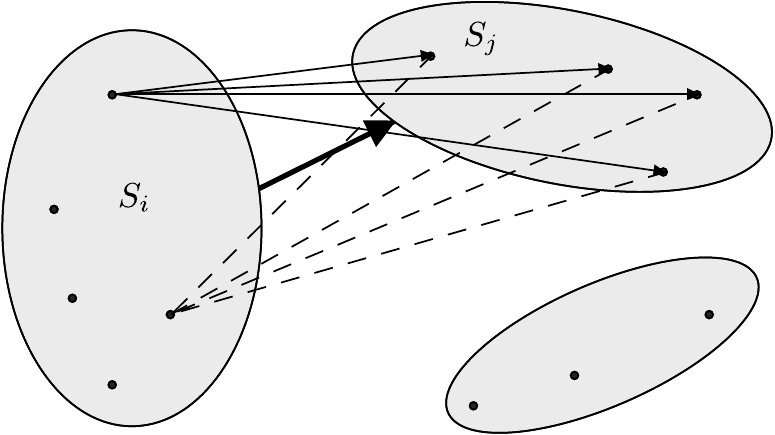}
    \caption{Illustration of the lumpability condition.}\label{fig:lump}
\end{figure}

\subsection{Deterministic swapping}

Consider a general quantum repeater with $n$ segments (which is not necessarily a power of two) with deterministic
swapping. In a previous section we studied the properties of such a repeater with a Markov chain with $2^n$ states,
which correspond to $n$-digit binary numbers. We show that it is possible to compress this chain to the size of $n+1$
and it will still contain the complete information about the waiting time of the corresponding repeater. Because
swapping never fails, all states with the same number of ready segments are equivalent. Let $S_i$, $i = 0, 1, \ldots,
n$, be the set of all $n$-binary digits with exactly $i$ components equal to one. Obviously, these sets form a partition
of the quantum repeater's set of states. The transition probabilities read as
\begin{equation}
    \mathbf{P}(S_i \to S_j) = 
    \begin{cases}
        0 & j < i \\
        \binom{n-i}{j-i} p^{j-i} q^{n-j} & i \leqslant j
    \end{cases}.
\end{equation}
We claim that the average waiting time of the repeater that started in the state $S_i$ is given by
\begin{equation}
    \overline{K}(i) = \sum^{n-i}_{l=1} (-1)^{l+1} \binom{n-i}{l} \frac{1}{1-q^l},
\end{equation}
which coincides with Eq.~\eqref{eq:KId} for any subset $I \subseteq [n]$ with $|I| = i$. We need to check that these
functions satisfy the linear equations
\begin{equation}
    1 + \sum^n_{j=i} \mathbf{P}(S_i \to S_j) \overline{K}(j) = \overline{K}(i).
\end{equation}
This can be easily done with the help of the following symbolic relation:
\begin{equation}
    \sum^{n-i}_{j=0} \,\,\, \sum^{n-i-j}_{l=1} = \sum^{n-i}_{l=1} \,\,\, \sum^{n-i-l}_{j=0}.
\end{equation}
In this case the original Markov chain had a lot of symmetry, which was exploited to reduce its size from exponential
with respect to the number of segments to a linear size. In the next subsection we show that a reduction in size is also
possible with less symmetric chains, though to a lesser degree.

\subsection{Doubling}

Consider a quantum repeater with non-deterministic swapping and a power-of-two number of segments $n = 2^d$ with the
standard doubling scheme. Before we demonstrated that this scheme requires a number of $N_n = 2^n$ states. This size is
better than the number of all states, $\mathcal{O}(2.61\ldots^n)$, but it is still impractical for $n > 8$. We now show
that we can describe this scheme by a Markov chain with $\mathcal{O}(1.34\ldots^n)$ states. This is still exponential
with respect to the number of segments, but it allows one to compute the exact rate of a quantum repeater with up to 32
segments.

\begin{figure}
    \includegraphics{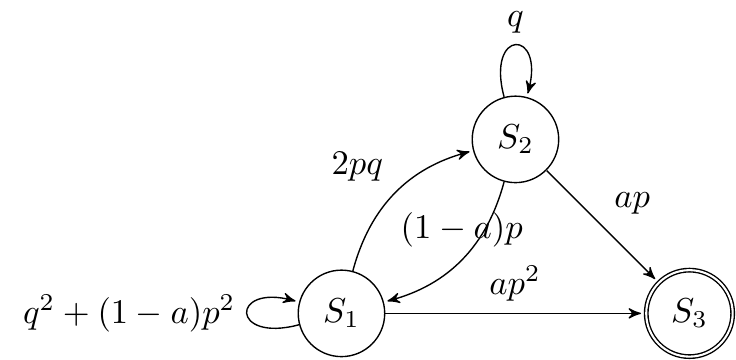}
    \caption{The compressed Markov chain of a 2-segment repeater.}\label{fig:M2C}
\end{figure}

Let us start with the simplest, 2-segment repeater. The doubling scheme in this case requires four states: $s_1 = 00$,
$s_2 = 01$, $s_3 = 10$ and $s_4 = \overline{11}$. Because the two halves of the repeater are independent in this scheme,
the states $01$ and $10$ are equivalent and can be grouped into one class. So, instead of four states, three are enough:
$S_1 = \{00\}$, $S_2 = \{01, 10\}$ and $S_3 = \{\overline{11}\}$. The corresponding Markov chain is shown in
Fig.~\ref{fig:M2C}. The 4-segment repeater can thus be described by nine states $S_i S_j$, $i, j = 1, \ldots, 3$ instead
of sixteen $s_i s_j$, $i, j = 1, \ldots, 4$. But we can go even further: the states $S_i S_j$ for $i \not= j$ are
equivalent and can be combined into one group, so six states are enough. These states are $\{S_1 S_1\}$, $\{S_1 S_2, S_2
S_1\}$, $\{S_1 S_3, S_3 S_1\}$, $\{S_2 S_2\}$, $\{S_2 S_3, S_3 S_2\}$, and $\{S_3 S_3\}$. This process is repeated on
higher levels. Instead of the recursive relation $N_{2n} = N^2_n$, we now have the relation 
\begin{equation}
    N_{2n} = \frac{1}{2}N_n(N_n + 1).
\end{equation}
There is no exact expression for the numbers defined by this recursive rule, but they can be well approximated as $N_n
\approx 2 \cdot 1.34\ldots^n$ \footnote{\texttt{https://oeis.org/A007501}}. For $n = 32$ the original Markov chain would
have more than $4 \cdot 10^9$ states, while the new, compressed chain has only $26796$ states. Solving the corresponding
system of linear equations requires around 20 Gb of RAM and takes approximately one minute of time on a modern 8-core
processor. Direct treatment of the original chain is intractable with the current technology.

\section{Verification}

Here we present a simple algorithm for Monte-Carlo simulation of a quantum repeater with a power-of-two number of
segments, $n = 2^d$, and a doubling scheme, which will allow us to validate our analytical results obtained before.
First, we need a function that simulates a single entanglement distribution over the quantum repeater. Its input
parameters are the level $d$, probabilities $p$ and $a$ and the boolean parameter $cc$ which says whether to include
classical communication or not. The output is the number of steps it took to distribute entanglement in this concrete
run, so the output will vary from run to run.

\begin{algorithmic}
    \Function{K}{$d$, $p$, $a$, $cc$}
        \State $k \gets 0$ \\
        \If {$d = 0$}
            \While {\textrm{true}}
                \State $k \gets k + 1$ \\
                \If {\Call{rand}{} $< p$}
                    \State \textrm{break}
                \EndIf
            \EndWhile
        \Else
            \While {\textrm{true}}
                \State $k_1 \gets $ \Call{K}{$d-1$, $p$, $a$, $cc$}
                \State $k_2 \gets $ \Call{K}{$d-1$, $p$, $a$, $cc$} \\
                \State $k \gets k + \max(k_1, k_2)$ \\
                \If {\Call{rand}{} $< a$}
                    \State \textrm{break}
                \Else
                    \If {$cc$}
                        \State $k \gets k + 2^{d-1}$
                    \EndIf
                \EndIf
            \EndWhile
        \EndIf \\
        \State \textbf{return} $k$
    \EndFunction
\end{algorithmic}

The function RAND here is a generator of pseudorandom numbers, uniformly distributed in the interval $[0, 1]$. The
function K implements the random variable $K$. To obtain the average waiting time, we need to call this
function many times and compute the average value of its output. One can also compute the variance and other higher
moments. Using this simple function, it is possible to verify that our analytical results are in very good agreement
with the results of the simulation.

\begin{figure}[ht]
    \includegraphics{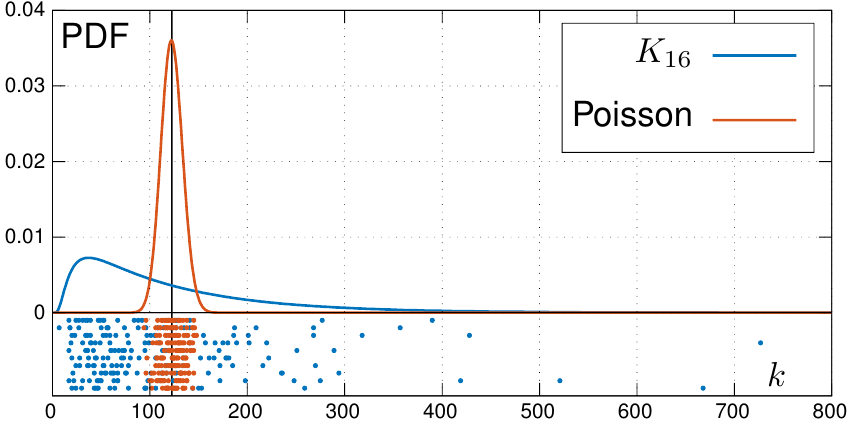}
    \caption{The PDF of $K_{16}$ (blue) with $p = a = 0.5$ vs. the Poisson PDF (red) with the same average value
    $\overline{K}_{16}$ (vertical).}\label{fig:qr-vs-P}
\end{figure}

This function can also be used to demonstrate the large variance of the waiting time. The results of calling it 200
times (10 series of experiments with 20 runs each) for a 16-segment quantum repeater are shown in the lower part of
Fig.~\ref{fig:qr-vs-P}. For comparison, we also plotted the same number of steps generated with a Poisson random number
generator with the same average. We can see the dramatic difference in the behavior of the two random variables. The
Poisson variable is compactly located around its average. On the other hand, the waiting time varies widely around the
average, and even in such a small number of experiments one can expect much shorter or much longer waiting times. The
short times are not problematic, but long waiting times can lead to failures of memory-based components. The upper half
of the figure shows the distributions of these random variables. The Poisson distribution has a well-known bell shape,
while the waiting time distribution significantly deviates from this shape. As a consequence, its average value alone is
not sufficient to give an acceptable characteristic of the waiting time random variable.

\section{Large repeaters}

The Markov chain approach developed above can give an exact result for small repeaters ($n \sim 2 - 4$), a numerical
result for medium-size repeaters ($n \sim 4 - 32$), but it is intractable for large repeaters ($n > 32$). For the latter
class of repeaters, we develop a family of approximations, where we can trade simplicity for accuracy and vice versa.

\begin{figure*}
    \centering
    \begin{subfigure}[b]{0.49\textwidth}
        \includegraphics[scale=0.95]{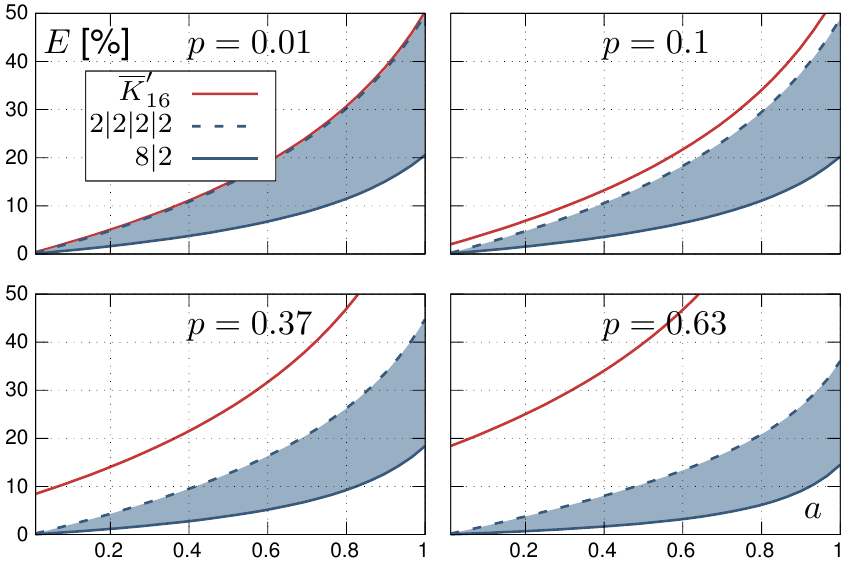}
    \end{subfigure}
    \begin{subfigure}[b]{0.49\textwidth}
        \includegraphics[scale=0.95]{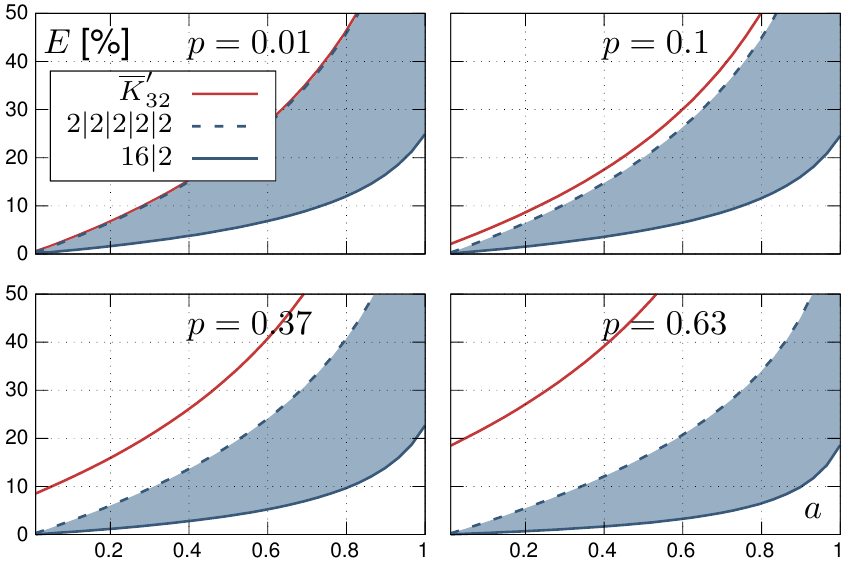}
    \end{subfigure}
   \caption{The relative error of different approximations as a function of $a$ for some values of $p$, 
   with $n=16$ (left) and $n=32$ (right) segments.}\label{fig:ea}
\end{figure*}

Generally, approximations are not necessary if we can compute exact values, but they are unavoidable when we cannot
obtain the result precisely. A simple way to get an approximation in the power-of-two case $n=2^d$ is to consider an
$n$-segment repeater as a $n/2$-segment one where each segment of the new repeater is a pair of the old repeater's
segments. Each new segment can be assigned an effective probability $p' = 1/\overline{K}_2$, and we have an approximate
relation 
\begin{equation}
    \overline{K}_n(p, a) \approx \overline{K}_{\frac{n}{2}}\left(\frac{1}{\overline{K}_2(p, a)}, a\right). 
\end{equation}
This scheme can be repeated and finally we have an approximation 
\begin{equation}
    \overline{K}_{2^d}(p, a) \approx \overline{K}_2\left(\frac{1}{\overline{K}_2\left(\frac{1}{\overline{K}_2(\ldots)}, a\right)}, a\right), 
\end{equation}
where on the right-hand side there are $d$ nested $\overline{K}_2$'s. A similar approach is used in
Ref.~\cite{arXiv.1811.01080} and also $\overline{K}'_n$ in Eq.~\eqref{eq:k2} is based on this kind of approach, however,
with the lowest-level waiting time approximated for small $p$. If $\overline{K}_2$ is the only average waiting time that
we know, then this approximation is the best that we can construct. If we know other averages $\overline{K}_4$,
$\overline{K}_8$, \ldots, then we can construct other, better approximations. For example, we have
\begin{equation}
  \overline{K}_{16}(p, a) \approx \overline{K}_8\left(\frac{1}{\overline{K}_2(p, a)}, a\right) 
  \approx \overline{K}_4\left(\frac{1}{\overline{K}_4(p, a)}, a\right), \nonumber
\end{equation}
and several others, where $\overline{K}$'s subindices are $2|8$, $2|4|2$, $2|2|4$, $2|2|2|2$ and $4|2|2$. To compare the
quality of different approximations, we shall use the quantity
\begin{equation}
    E = 100 \left|\frac{\text{approx.}}{\text{exact}} - 1\right|,
\end{equation}
which is the relative error measured in percents. In Fig.~\ref{fig:ea} we plot these approximations for some practical
values of $p = e^{-L_0/L_{\mathrm{att}}}$ with $L_{\mathrm{att}} = 22\mathrm{km}$ (corresponding to $L_0 = 100$, $50$,
$22$, and $10$ km, respectively, assuming deterministic local state preparations). As Fig.~\ref{fig:ea} illustrates, the
least precise approximation is $2|\ldots|2$, the most accurate is $n/2|2$ (where we introduce effective probabilities
only at the lowest level), and all other variants are in between these two (in Fig.~\ref{fig:ea} this area is shaded).
The approximation given by Eq.~\eqref{eq:k2} is the least precise. As one can see, the quality of the best approximation,
$n/2|2$, degrades not so quickly as that of the other approximations.

We considered approximations to $\overline{K}_{16}$, which we can compute exactly, only to demonstrate the quality of
different approximations. For $n \geqslant 64$ the problem of computing $\overline{K}_n$ with the Markov chains method
is intractable, so we have to use approximations. As we have shown, the best approximation is obtained by nesting exact
values of $\overline{K}_{n'}$ with as large $n'$ as possible. Moreover, inner $n'$ should not exceed outer ones. The
best available approximation to $\overline{K}_{1024}$ reads as
\begin{equation}
    \overline{K}_{1024}(p, a) \approx \overline{K}_{32}\left(\frac{1}{\overline{K}_{32}(p, a)}, a\right).
\end{equation}
The inverse, $1/\overline{K}_{1024}(p, a)$, gives us the best available approximation to the distribution rate.
Generally, increasingly better lower bounds on the (raw) repeater rates can be obtained by our method. Depending on
$L_0$, which typically is $10$ -- $100$ km, $n = 1024$ segments cover the distance $L = n L_0 > 10000$ km. Using
free-space satellite links to bridge these repeaters, it is probably unnecessary to use more than $1024$ segments, and
thus we will never need more than two nesting levels in the approximations of this form.

\section{Conclusions}

In summary, in addition to satellite-based long-distance quantum communication links, fiber-based quantum repeaters are
necessary for creating reliable large-scale quantum communication networks. Memory-based quantum repeaters are starting
to be experimentally realized. Up to now knowledge about the random waiting time of a general quantum repeater has been
incomplete and imprecise. Here we completely solve the problem of its probability distribution function and demonstrate
that, contrary to common belief, the waiting time cannot be accurately characterized by its average value alone. Our
approach is applicable to general quantum repeaters, including repeaters with finite memory effects and waiting times
spent for classical communication, and it allows one to obtain the full probability distribution of the waiting time,
which to a good approximation in a certain regime turns out to behave like a geometric progression. We expect that
precise knowledge about the waiting-time statistics also has a significant impact on the treatment of errors in a
quantum repeater such as those arising from memory dephasing. Additional probabilistic entanglement manipulations, such
as entanglement distillation for suppressing the propagation of errors, can also be incorporated into a repeater rate
analysis using our formalism. In the QKD context, in order to determine the secret key rate in a quantum repeater
system, the raw rates must be calculated and errors must be included via the secret key fraction. In our work, putting
it in this context, the focus has been on repeater raw rates which are highly non-trivial to compute in a general
quantum repeater with arbitrary, probabilistic entanglement swapping. Not only did we solve a long-standing problem,
eventually allowing for the assessment and creation of truly robust and reliable quantum devices, we also presented a
method for compressing Markov chains that can be useful for other applications as well.

\end{document}